\documentclass{article}

     \usepackage[final,nonatbib]{neurips_2020}

\usepackage[utf8]{inputenc} 
\usepackage[T1]{fontenc}    
\usepackage{hyperref}       
\usepackage{url}            
\usepackage{booktabs}       
\usepackage{amsfonts}       
\usepackage{nicefrac}       
\usepackage{microtype}      

\usepackage{enumitem}

\usepackage{xcolor}
\usepackage{amsthm}
\usepackage{amsmath}

\usepackage{bm}

\newtheorem{theorem}{Theorem}

\newtheorem{definition}[theorem]{Definition}
\newtheorem{proposition}[theorem]{Proposition}
\newtheorem{lemma}[theorem]{Lemma}

\usepackage{algorithm}
\usepackage[noend]{algpseudocode}

\newcommand{\ucb}{\mathop{}\!\mathrm{ucb}}

\usepackage{tcolorbox}
\usepackage{xfrac}
\usepackage{dsfont}

\usepackage{caption}
\title{Contextual Games: \\ Multi-Agent Learning with Side Information}

\author{%
  Pier Giuseppe Sessa \\
  ETH Z\"urich\\
  \texttt{sessap@ethz.ch} \\
\And
Ilija Bogunovic \\
  ETH Z\"urich\\
  \texttt{ilijab@ethz.ch} \\
  \And
  Andreas Krause \\
  ETH Z\"urich\\
  \texttt{krausea@ethz.ch} \\
    \And
Maryam Kamgarpour \\
  ETH Z\"urich\\
  \texttt{maryamk@ethz.ch} \\
}
\begin{document}

\maketitle

\begin{abstract}
We formulate the novel class of \emph{contextual games}, a type of repeated games driven by contextual information at each round. By means of \emph{kernel-based} regularity assumptions, we model the correlation between different contexts and game outcomes and propose a novel online (meta) algorithm that exploits such correlations to minimize the \emph{contextual regret} of individual players. We define game-theoretic notions of \emph{contextual Coarse Correlated Equilibria} (c-CCE) and \emph{optimal contextual welfare} for this new class of games and show that c-CCEs and optimal welfare can be approached whenever players' contextual regrets vanish. Finally, we empirically validate our results in a traffic routing experiment, where our algorithm leads to better performance and higher welfare compared to baselines that do not exploit the available contextual information or the correlations present in the game.\looseness=-1
\end{abstract}

\vspace{-0.5em}
\section{Introduction}
\vspace{-0.5em}
Several important real-world problems, ranging from economics, engineering, and computer science involve multiple interactions of self-interested agents with coupled objectives. They can be modeled as \emph{repeated games} and have received recent attention due to their connection with learning (e.g., \cite{cesa-bianchi_prediction_2006}).\looseness=-1

An important line of research has focused, on the one hand, on characterizing game-theoretic equilibria and their efficiency and, on the other hand, on deriving fast learning algorithms that converge to equilibria and efficient outcomes.
Most of these results, however, are based on the assumption that the players always face the exact same game, repeated over time. While this leads to strong theoretical guarantees, it is often unrealistic in practical scenarios: In routing games~\cite{roughgarden2007}, for instance, the agents’ travel times and hence the ‘rules’ of the game are governed by many time-changing factors such as network’s capacities, weather conditions, etc. Often, players can observe such factors, and hence could take better decisions depending on the circumstances.

Motivated by these considerations, we introduce the new class of \emph{contextual games}. Contextual games define a more general class of repeated games described by different contextual, or side,  information at each round, denoted also as \emph{contexts} in analogy with the bandit optimization literature (e.g., \cite{langford2008}). Importantly, in contextual games players can observe the current context before playing an action, which allows them to achieve better performance, and converge to stronger notions of equilibria and efficiency than in standard repeated games.  

\textbf{Related work.} Learning in repeated \emph{static} games has been extensively studied in the literature. The seminal works \cite{hannan1957,hart2000} show that simple \emph{no-regret} strategies for the players converge to the set of Coarse Correlated Equilibria (CCEs), while the \emph{efficiency} of such equilibria and learning dynamics has been studied in \cite{blum2008,roughgarden2015}. Exploiting the static game structure, moreover, \cite{syrgkanis2015,foster2016} propose faster learning algorithms, and a long array of works (e.g., \cite{Singh2000,Bowling2005,Balduzzi2018}) study convergence to Nash equilibria.  
Learning in \emph{time-varying} games, instead, has been recently considered \cite{duvocelle2018}, where the authors show that dynamic regret minimization allows players to track the sequence of Nash equilibria, provided that the stage games are monotone and slowly-varying. Adversarially changing zero-sum games have also been studied \cite{cardoso19a}, with convergence guarantees to the Nash equilibrium of the time-averaged game. Our contextual games model is fundamentally different than \cite{duvocelle2018,cardoso19a} in that we assume players observe the current context (and hence have prior information on the game) before playing. This leads to new equilibria and a different performance benchmark, denoted as \emph{contextual regret}, described by the best \emph{policies} mapping contexts to actions. 
Perhaps closer to ours is the setup of stochastic (or Markov) games \cite{Shapley1095}, at the core of multi-agent reinforcement learning (see \cite{Busoniu2010} for an overview). There, players observe the \emph{state} of the game before playing but, differently from our setup, the evolution of the state depends on the actions chosen at each round. This leads to a nested game structure, which requires significant computational power and players' coordination to compute equilibrium strategies via backward induction \cite{Greenwald2003,Dermed2009}. Instead, we consider arbitrary contexts' sequences (potentially chosen by an adversarial Nature) and show that efficient algorithms converge to our equilibria in a decentralized fashion.\looseness=-1

From a single player's perspective, a contextual game can be reduced to a special \emph{adversarial} contextual bandit problem \cite[Chapter 4]{bubeck2012}, for which several no-regret algorithms exist. All such algorithms, however, rely on high-variance estimates for the rewards of non-played actions and thus their performance degrades with the number of actions available. 
A fact not exploited by these algorithms is that in a contextual game {\em similar contexts and game outcomes likely produce similar rewards} (e.g., in a routing game, similar network capacities and occupancy profiles lead to similar travel times). We encode this fact using kernel-based regularity assumptions and (similarly to \cite{sessa2019noregret} in non-contextual games) show that exploiting these assumptions, and additionally observing the past opponents' actions, players can achieve substantially improved performance
compared to using standard bandit algorithms. For instance, for $K$ actions and adversarially chosen contexts from a finite set $\mathcal{Z}$, the bandit $\mathcal{S}$-\textsc{Exp3}\cite{bubeck2012} incurs $\mathcal{O}(\sqrt{T K |\mathcal{Z}| \log K})$ contextual regret, while our approach leads to a $\mathcal{O}(\sqrt{T |\mathcal{Z}| \log K} + \gamma_T \sqrt{T})$ guarantee, where $\gamma_T$ is a sample-complexity parameter describing the degrees of freedom in the player's reward function. For commonly used kernels, this results in a sublinear regret bound that grows only {\em logarithmically} in $K$. Moreover, when contexts are stochastic and private to a player, we obtain a $\mathcal{O}(\sqrt{T\log K} + \gamma_T \sqrt{T})$ pseudo-regret bound. This bound should be compared to the $\mathcal{O}(\sqrt{T K \log K})$ pseudo-regret of \cite{balseiro2019} which --unlike us-- assumes observing the rewards for non-revealed contexts, and the $\mathcal{O}(\sqrt{c T K \log K})$ pseudo-regret of \cite{neu2020}, which assumes known contexts distribution and a linear dependence of rewards on contexts in $\mathbb{R}^c$.\looseness=-1

\textbf{Contributions.} We formulate the novel class of \emph{contextual games}, a type of repeated games characterized by (potentially) different contextual information available at each round.
\vspace{-0.5em}
\begin{itemize}[itemsep= 0pt,topsep=1pt,leftmargin = 1.0em]
\item[-] We identify the \emph{contextual regret} as a natural benchmark for players' individual performance, and propose novel online algorithms to play contextual games with no-regret. Unlike existing contextual bandit algorithms, our algorithms exploit the correlation between different game outcomes, modeled via kernel-based regularity assumptions, 
yielding improved performance.\looseness=-1
\item[-] We characterize equilibria and efficiency of contextual games, defining the new notions of \emph{contextual Coarse Correlated Equilibria (c-CCE)} and \emph{optimal contextual welfare}. We show that c-CCEs and contextual welfare can be approached in a decentralized fashion whenever players minimize their contextual regrets, thus recovering important game-theoretic results for our larger class of games.\looseness=-1  
\item[-] We demonstrate our results in a repeated traffic routing application. Our algorithms effectively use the available contextual information (network capacities) to minimize agents' travel times and converge to more efficient outcomes compared to other baselines that do not exploit the observed contexts and/or the correlations present in the game.
\end{itemize}

\vspace{-0.8em}
\section{Problem Setup}
\vspace{-0.8em}
We consider repeated interactions among $N$ agents, or players.
At every round, each player selects an action and receives a payoff that depends on the actions chosen by all the players as well as the \emph{context} of the game at that round. 
More formally, we let $\mathcal{Z}$ represent the (potentially infinite) set of possible contexts, and $\mathcal{A}^i$ be the set of actions available to player $i$. Then, we define $r^i: \bm{\mathcal{A}} \times \mathcal{Z} \rightarrow [0,1]$ to be the reward function of each player $i$, where $\bm{\mathcal{A}} := \mathcal{A}^1 \times \dots \times \mathcal{A}^N$ is the joint action space. Importantly, we assume $r^i$ is \emph{unknown} to player $i$.
With the introduced notation, a repeated \emph{contextual game} proceeds as follows. At every round $t$: 
\begin{itemize}[noitemsep,topsep=-4pt]
  \setlength\itemsep{0.05em}
\item Nature reveals context $z_t$ 
\item Players observe $z_t$ and, based on it, each player $i$ selects action $a_t^i \in \mathcal{A}^i$, for $i = 1,\ldots N$.
\item Players obtain rewards $r^i\big(a_t^i, a^{-i}_t, z_t\big)$, $i=1,\ldots, N$.
\end{itemize}

Moreover, as specified later, player $i$ receives feedback information at the end of each round that it can use to improve its strategy.
Let $\Pi^i$ be the set of all policies $\pi:  \mathcal{Z} \rightarrow \mathcal{A}^i$, mapping contexts to actions. 
After $T$ game rounds, the performance of player $i$ is measured by the \emph{contextual regret}:
{
\begin{equation}\label{eq:contextual_regret}
R_c^i(T) = \max_{\pi \in \Pi^i}  \sum_{t=1}^T r^i \big( \pi(z_t), a^{-i}_t, z_t\big) -   \sum_{t=1}^T  r^i\big(a_t^i, a^{-i}_t, z_t\big) \, .
\end{equation}
}
The contextual regret compares the cumulative reward obtained throughout the game with the one achievable by the \emph{best fixed policy} in hindsight, i.e., had player $i$ known the sequence $\{z_t, a^{-i}_t\}_{t=1}^T$ of contexts and opponents' actions ahead of time, as well as the reward function $r^i(\cdot)$. Crucially, $R_c(T)$ sets a stronger benchmark than competing only with the best fixed action $a \in \mathcal{A}^i$ and captures the fact that players should use the revealed context information to improve their performance. A strategy is {\em no-regret} for player $i$ if $R_c^i(T)/T \rightarrow 0$ as $T\rightarrow \infty$.

Contextual games generalize the class of standard (non-contextual) repeated games, allowing the game to change from round to round due to a potentially different context $z_t$ (we recover the standard repeated games setup and regret definition by assuming $z_t = z_0$ for all $t$). In Section~\ref{sec:equilibria_and_efficiency} we define new notions of \emph{equilibria} and \emph{efficiency} for such games and show that the contextual regret defined in \eqref{eq:contextual_regret}, besides measuring individual players' performance, has a close connection with game equilibria and efficiency. First, however, motivated by these considerations, we focus on the individual perspective of a generic player $i$ and seek to derive suitable no-regret strategies. In this regard, the algorithms and guarantees presented in the next section do not rely on the other players complying with any pre-specified rule, but consider the worst case over the opponents' actions $a_t^{-i}$ (also, potentially chosen as a function of the observed game data). To simplify notation, we denote player~$i$'s reward function with $r(\cdot)$, unless otherwise specified.\looseness=-1

\vspace{-0.8em}
\paragraph{Regularity assumptions.} Even with a single known context, achieving no-regret is impossible unless we make further assumptions on the game \cite{cesa-bianchi_prediction_2006}. We assume the action set $\mathcal{A}^i$ is finite with $|\mathcal{A}^i| = K$.
We consider a generic set $\mathcal{Z}\subseteq \mathbb{R}^c$ and make no assumptions on how contexts are generated (they could be adversarially chosen by Nature, possibly as a function of past game rounds). In Section
~\ref{sec:stochastic_contexts}, however, we consider a special case where contexts are sampled i.i.d. from a static distribution $\zeta$. Our next regularity assumptions concern the reward function $r(\cdot)$.\\
Let $\mathcal{D}:= \bm{\mathcal{A}} \times \mathcal{Z}$. We assume the unknown function $r(\cdot)$ has a bounded norm in a Reproducing Kernel Hilbert Space (RKHS) associated with a positive-semidefinite kernel function $k:\mathcal{D} \times \mathcal{D} \rightarrow [-1,1]$. Kernel $k(\cdot, \cdot)$ measures the similarity between two different context-action pairs, and the norm $\|r\|_k = \sqrt{\langle r,r\rangle_k}$ measures the smoothness of $r$ with respect to $k$. This is a standard non-parametric assumption used, e.g., in Bayesian optimization \cite{srinivas2009gaussian} and recently exploited to model correlations in repeated games \cite{sessa2019noregret}. It encodes the fact that similar action profiles (e.g., similar network's occupancy profiles in a routing game) lead to similar rewards (travel times), and allows player $i$ to generalize experience to non-played actions and, in our case, unseen contexts. Popularly used kernels include polynomial, Squared Exponential (SE), and Matérn kernels \cite{rasmussen2006gaussian}, while composite kernels \cite{krause2011contextual} can also be used to encode different dependences of $r$ on $a^i$, $a^{-i}$, and context $z$.

\vspace{-1em}
\paragraph{Feedback model.} We assume player $i$ receives a \emph{noisy} bandit observation $\tilde{r}_t = r(a_t^i, a_t^{-i}, z_t) + \epsilon_t$ of the reward at each round, where $\epsilon_t$ is  $\sigma$-sub-Gaussian (i.e., $\mathbb{E}[\exp(\alpha \epsilon_t)] \leq \exp(\alpha^2\sigma^2/2)$, $\forall \alpha \in \mathbb{R}$) and independent over time.
Moreover, we assume, similar to \cite{sessa2019noregret}, that at the \emph{end} of each round, player $i$ also observes the actions $a_t^{-i}$ chosen by the other players. The latter assumption will allow player $i$ to achieve improved performance compared to the standard bandit feedback.
In some applications (e.g., aggregative games such as traffic routing), it is only sufficient to observe an aggregate function of $a_t^{-i}$.\looseness=-1

\vspace{-0.5em}
\section{Algorithms and Guarantees}\label{sec:algorithms}
\vspace{-0.5em}
From the perspective of player $i$, playing a contextual game corresponds to an \emph{adversarial contextual bandit problem} (see, e.g., \cite[Chapter 4]{bubeck2012}) where, at each round, context $z_t$ is revealed, an adversary picks a reward function $r_t(\cdot,z_t): \mathcal{A}^i \rightarrow [0,1]$ and player $i$ obtains reward $r_t(a_t^i,z_t)$. Therefore, player $i$ could in principle use existing adversarial contextual bandits algorithms to achieve no-regret. Such algorithms come with different regret guarantees depending on the assumptions made but, importantly, incur a regret which scales poorly with the size of the action space $\mathcal{A}^i$. This is because they use high-variance estimators to estimate the rewards of non-played actions, i.e., the so-called \emph{full information} feedback. Also, some of these algorithms assume parametric (e.g., linear \cite{neu2020}) dependence of the rewards $r_t(\cdot, z_t)$ on the context $z_t$ and hence cannot deal with our general game structure.\looseness=-1 

Instead, we exploit the fact that in a contextual game the rewards obtained at different times are \emph{correlated} through the reward function $r(\cdot)$ (i.e., contextual games correspond to the specific contextual bandit problem where $r_t(\cdot,z_t) = r(\cdot, a_t^{-i},z_t)$ for all $t$). This fact, together with our feedback model, allows player $i$ to use past game data to obtain (with increasing confidence) an estimate of the reward function $r$ and use it to emulate the full-information feedback. 

\setlength\belowdisplayskip{4pt}
\setlength\abovedisplayskip{3pt}
\vspace{-0.8em}
\paragraph{RKHS regression.} Using past game data $\{a_\tau^i, a_\tau^{-i}, z_\tau, \tilde{r}_\tau\}_{\tau=1}^{t}$, standard kernel ridge regression \cite{rasmussen2006gaussian} can be used to compute posterior mean and corresponding variance estimates of the reward function $r(\cdot)$. For any $x = (a, a^{-i}, z) \in \mathcal{D}$, and regularization parameter $\lambda > 0$, they can be obtained as:\looseness=-1
\begin{align}
	\mu_{t}(x) = \mathbf{k}_t(x)^T\big(\mathbf{K}_t + \lambda \mathbf{I}_t \big)^{-1} \mathbf{y}_t  \, , \quad 
	\sigma_{t}^2(x) = k(x,x)   -\mathbf{k}_t(x)^T \big(\mathbf{K}_t + \lambda \mathbf{I}_t \big)^{-1} \mathbf{k}_t(x) \, , \label{eq:posterior_mean_and_variance}
\end{align}
where $\mathbf{k}_t(x) = \big[k\big(x_j,x\big)\big]_{j=1}^t$, $\mathbf{y}_t= \big[\tilde{r}_j \big]_{j=1}^t$, and $\mathbf{K}_t = \big[k \big(x_j,x_{j'}\big)\big]_{j,j'}$ is the kernel matrix.
Moreover, such estimates can be used to build the upper confidence bound function:
\begin{align} 
    \ucb_{t}(\cdot) &:= \min \{  \mu_t(\cdot) + \beta_{t} \sigma_{t}(\cdot) \: , 1\} \, , \label{eq:ucb}
\end{align}
where $\beta_t$ is a confidence parameter, and the function is truncated at $1$ since $r(x)\in [0,1]$ for all $x\in \mathcal{D}$. A standard result from \cite{srinivas2009gaussian} shows that $\beta_t$ can be chosen such that $r(x) \in [\mu_t(x) + \beta_{t} \sigma_{t}(x), \mu_t(x) - \beta_{t} \sigma_{t}(x)]$ with high probability for any $x \in \mathcal{D}$ (see Lemma~\ref{lemma:conf_lemma} in the Appendix). The function $\ucb_t(\cdot)$ hence represents an optimistic estimate of $r(\cdot)$ and can be used by player $i$ to emulate the full-information feedback. We outline our (meta) algorithm \textsc{c.GP-MW} in Algorithm~\ref{alg:algorithm1}.

\setlength{\textfloatsep}{15pt}
\begin{algorithm}[t!]
\caption{ The \textsc{c.GP-MW}  (meta) algorithm}\label{alg:algorithm1}
\begin{algorithmic}[1]
\vspace{-0.2em}
    \Require Finite set $\mathcal{A}^i$ of $K$ actions, kernel $k$, learning rates $\{ \eta_t \}_{t\geq 0}$, confidence levels $\{ \beta_t\}_{t \geq 0}$.
    \For {$t = 1, \ldots, T$}
    
 \noindent{\color{blue} { \footnotesize\ttfamily  /* Nature chooses context $z_t$  \hfill /* }}    
        \State  Observe context $z_t$  
        \State Compute distribution $p_t(z_t) \in \Delta^K$ using: $z_t$, $\eta_t$, and $\{\ucb_\tau(\cdot), a_\tau^{-i}, z_\tau\}_{\tau=1}^{t-1}$.
        \State Sample action $a^i_t \sim p_t(z_t)$
        
        \noindent{\color{blue} { \footnotesize\ttfamily  /* Simultaneously, opponents choose their actions $a^{-i}_t$  \hfill /* }}    
        \State Observe noisy reward $\tilde{r}_t$ and opponents' actions $a^{-i}_t$ \Comment{$\tilde{r}_t = r(a^i_t, a^{-i}_t, z_t) + \epsilon_t$}
        \State  Use the observed data to update $\ucb_t(\cdot)$ according to \eqref{eq:posterior_mean_and_variance} and \eqref{eq:ucb}.
    \EndFor  
    \vspace{-0.3em}
\end{algorithmic}
\end{algorithm}

 \textsc{c.GP-MW} extends and generalizes the recently proposed \textsc{GP-MW}~\cite{sessa2019noregret} algorithm to play repeated games, to the case where contextual information is available to the players and the goal is to compete with the optimal policy in hindsight. At each time step, after observing context $z_t$ the algorithm computes a distribution $p_t(z_t)\in \Delta^{K}$, where $\Delta^K$ is the $K$-dimensional simplex, and samples an action from it. At the same time, the algorithm uses the observed game data to construct upper confidence bound functions of the player's rewards using \eqref{eq:ucb}. Such functions, together with the observed context $z_t$ are used to compute the distribution $p_t(z_t)$ at each round. Note that \textsc{c.GP-MW} is a well-defined algorithm after we specify the rule used to compute $p_t(z_t)$ (line 3 of Algorithm~\ref{alg:algorithm1}). We left such rule unspecified, as we will specialize it to different settings throughout this section.

The regret bounds obtained in this section depend on the so-called
\emph{maximum information gain} \cite{srinivas2009gaussian} about the unknown function $r(\cdot)$ from $T$ noisy observations, defined as: 
\begin{equation*}
    \gamma_T := \max_{{\lbrace x_t \rbrace}_{t=1}^{T}} 0.5\log \det(\mathbf{I}_T + \mathbf{K}_T / \lambda).
\end{equation*}
This quantity represents a sample-complexity parameter which,  importantly, for popularly used kernels does not grow with the number of actions $K$ but only with the dimension $d$ of the domain $\mathcal{D}$.  It can be bounded analytically as, e.g., $\gamma_T \leq \mathcal{O}(d\log T)$ and $\gamma_T \leq \mathcal{O}(\log(T)^{d+1})$ for squared exponential and linear kernels, respectively  \cite{srinivas2009gaussian}. Moreover, we remark that although in the worst case $d$ grows linearly with the number of players $N$, in many applications (such as the traffic routing game considered in Section~\ref{sec:experiments}) the reward function $r(\cdot)$ depends only on some aggregate function of the opponents' actions $a^{-i}$ and therefore $d$ is independent from the number of players in the game.

\vspace{-0.5em}
\subsection{Finite (small) number of contexts}\label{sec:finite_contexts}
\vspace{-0.5em}
When the context set $\mathcal{Z}$ is finite, a basic version of \textsc{c.GP-MW} achieves a high-probability regret bound of $\mathcal{O}(\sqrt{T|\mathcal{Z}|\log K} + \gamma_T \sqrt{T})$: We simply maintain a distribution $p_t(z)\in \Delta^{K}$ for each context $z$ and update it only when $z$ is observed, using the Multiplicative Weights (MW) method \cite{littlestone1994}. We formally introduce and study such strategy in Appendix~\ref{app:simple_case}. In the same setting, and with standard bandit feedback, the $\mathcal{S}$-\textsc{Exp3}\cite{bubeck2012} algorithm achieves regret $\mathcal{O}(\sqrt{T |\mathcal{Z}| K \log K})$, which has a worse dependence on the number of actions $K$. These regret bounds, however, are appealing only when the set $\mathcal{Z}$ has low cardinality, and become worthless otherwise.
Intuitively, this is because no information is shared across different contexts and each context is treated independently from each other.\looseness=-1

\vspace{-0.5em}
\subsection{Exploit similarity across contexts}
\vspace{-0.5em}
For large or even infinite $\mathcal{Z}$, we want to exploit the fact that similar contexts should lead to similar performance and take this into account when computing the action distribution $p_t(z_t)$. 
We capture this fact by assuming the optimal policy in hindsight $\pi^\star = \arg\max_{\pi \in \Pi^i} \sum_{t=1}^T r(\pi(z_t), a_t^{-i}, z_t)$ is $L_p$-Lipschitz:
{\setlength{\belowdisplayskip}{5pt} \setlength{\belowdisplayshortskip}{5pt}
\setlength{\abovedisplayskip}{-5pt} \setlength{\abovedisplayshortskip}{-5pt}
\begin{equation*}
|\pi^\star(z_1) -\pi^\star(z_2)| \leq L_p \|z_1 - z_2\|_1, \quad \forall z_1, z_2 \in \mathcal{Z} \,.
\end{equation*}
}
Moreover, we assume $\mathcal{Z}\subseteq [0,1]^c$ to obtain a scale-free regret bound, and that the reward function $r(\cdot)$ is $L_r$-Lipschitz with respect to the decision set $\mathcal{A}^i$, i.e., $| r(a_1, a^{-i}, z ) - r(a_2, a^{-i}, z )|  \leq L_r \|a_1 - a_2 \|_1 , \forall a_1,a_2 \in \mathcal{A}^i ,  \forall (a^{-i}, z)$, which is readily satisfied for most kernels \cite[Lemma 1]{de2012regret}.

These assumptions allow player $i$ to share information across different, but similar, contexts to improve the performance. This can be done by using the online Strategy~\ref{strategy} to compute $p_t(z_t)$ at each round (Line 3 in Algorithm~\ref{alg:algorithm1}).
\setlength{\textfloatsep}{15pt}
\makeatletter
\renewcommand{\ALG@name}{Strategy}
\makeatother
\begin{algorithm}[t!]
\caption{Exploit similarity across contexts}\label{strategy}
\begin{algorithmic}[1]
\vspace{-0.2em}
    \State \textbf{Set} Radius $\epsilon > 0$, $\mathcal{C}= \{z_1 \}$, and let $p_1(z_1)$ be the uniform distribution. 
    \For {$t = 2, \ldots, T$}
    \State Observe context $z_t$ and let $z_t'= \arg \min_{z \in \mathcal{C}} \|z_t - z \|_1$
        \If{$\|z_t - z_t' \|_1 > \epsilon$}
        \State Add $z_t$ to the set $\mathcal{C}$, set $z_t' =z_t$, and let $p_t(z_t)$ be the uniform distribution
        \Else
        \vspace{-1.2em}
    \begin{equation}\label{eq:GPMWrule_lipschitz}
p_t(z_t)[a] \propto \exp \Bigg( \eta_t \cdot  \sum_{\tau=1}^{t-1}  \ucb_\tau(a, a^{-i}_\tau, z_\tau)\cdot \mathds{1}\{z_\tau' = z_t' \} \Bigg)  \qquad a=1, \ldots, K \, .
\end{equation}
\EndIf
    \EndFor  
\end{algorithmic}
\vspace{-0.8em}
\end{algorithm}
Such strategy consists of building, in a greedy fashion as new contexts are revealed, an $\epsilon$-net \cite{krauthgamer04} of the context space $\mathcal{Z}$, similarly to the algorithm by \cite{hazan2007} for online convex optimization: At each time $t$, either a new L1-ball centered at $z_t$ is created or $z_t$ is assigned to the closest ball. In the latter case, $p_t(z_t)$ is computed via a MW rule using the sequence of $\ucb_\tau(\cdot)$ functions for those time steps $\tau < t$ that $z_\tau$ belongs to such ball. Note that Strategy~2 can also be implemented recursively, by maintaining a probability distribution for each new ball and updating only the one corresponding to the ball $z_t$ belongs to. The radius $\epsilon$ is a tunable parameter, which can be set as follows.\looseness=-2

{\setlength\belowdisplayskip{0pt}
\setlength\abovedisplayskip{2pt}
\begin{theorem}\label{thm:thm1}
Fix $\delta \in (0,1)$ and assume $\| r^i\|_{k}\leq B$, $\pi^\star$ is $L_p$-Lipschitz, and $r^i$ is $L_r$-Lipschitz in $\mathcal{A}^i$. If player $i$ plays according to \textsc{c.GP-MW} using Strategy~\ref{strategy} with $\lambda \geq 1$, $\beta_t = B + \sigma \lambda^{-1/2} \sqrt{2(\gamma_{t-1} + \log(2/\delta))}$, $\eta_t = 2\sqrt{\log K/\sum_{\tau=1}^{t}\mathds{1}\{z_\tau' = z_t'\}}$, and $\epsilon = (L_r L_p)^{-\frac{2}{c + 2}} T^{-\frac{1}{c+2}}$, then with probability at least $1- \delta$,
\begin{equation*}
R_c^i(T) \leq  2(L_r L_p)^\frac{c}{c + 2} \:  T^\frac{c+1}{c+2} \sqrt{\log K} + \sqrt{0.5 T \log(2/ \delta)}  + 4 \beta_T \sqrt{\gamma_T \lambda T} \, .
\end{equation*}
\vspace{-1em}
\end{theorem}
}
Compared to Section~\ref{sec:finite_contexts}, the obtained regret bound is now independent of the size of $\mathcal{Z}$, although its sublinear dependence on $T$ degrades with the contexts' dimension $c$. The additive $\mathcal{O}(\beta_T \sqrt{\gamma_T T})$ term represents the cost of learning the reward function $r(\cdot)$ online. Note that even if $r(\cdot)$ was {\em known}, and hence full-information feedback was available, the $\mathcal{O}(\sqrt{T |\mathcal{Z}|}) $ and $\mathcal{O}(T^\frac{c+1}{c+2})$ rates obtained so far are shown optimal in their respective settings, i.e., when $\mathcal{Z}$ is finite \cite{bubeck2012} or when the discussed Lipschitz assumptions are satisfied \cite{hazan2007}. An interesting future direction is to understand whether more refined bounds can be derived as a function of the contexts' sequence using adaptive partitions as proposed by \cite{slivkins2011}. In the next section we show that significantly improved guarantees are achievable when contexts are i.i.d. samples from a static distribution.

Finally, we remark that all the discussed computations are efficient, as they do not iterate over the set of policies $\Pi^i$ (which has exponential size). Improved regret bounds can be obtained if this requirement is relaxed, e.g., assuming a finite pool of policies \cite{auer2003}, or a value optimization oracle \cite{syrgkanis2016contextual}. We believe such results are complementary to our work and can be coupled with our RKHS game assumptions.\looseness=-1

\vspace{-0.5em}
\subsection{Stochastic contexts and non-reactive opponents}\label{sec:stochastic_contexts}
\vspace{-0.5em}
In this section, we consider a special case where contexts are i.i.d. samples from a static distribution $\zeta$, i.e., $z_t \sim \zeta$ for $t=1, \ldots,T$. Importantly, we consider the realistic case in which player $i$ does neither know, nor can sample from, such distribution.
Moreover, we focus on the setting where the opponents' decisions $a_t^{-i}$ are \emph{not} based on the current realization of $z_t$, but can only depend on the history of the game. Examples of such a setting are games where the context $z_t$ represents `private' information for player $i$ (e.g., in Bayesian games, $z_t$ can represent player $i$'s \emph{type} \cite{hartline2015}), or where $z_t$ is only relevant to player $i$ and hence the opponents have no reason to decide based on it.

In this case, we analyze the following strategy to compute $p_t(z_t)$ at each round (Line 3 in Algorithm~\ref{alg:algorithm1}):
\begin{equation}\label{eq:GPMWrule_v2}
p_t(z_t)[a] \propto \exp\Bigg( \eta_t \cdot \sum_{\tau=1}^{t-1}  \ucb_\tau(a, a^{-i}_\tau, z_t) \Bigg)   \qquad a=1. \ldots, K \, .
\end{equation}
Crucially, the distribution $p_t(z_t)$ is now computed using the {\em whole sequence} of past $\ucb_\tau$ functions, evaluated at context $z_t$, regardless of whether $z_t$ was observed in the past. 
 Hence, while according to rule~\eqref{eq:GPMWrule_lipschitz} -- and most of the MW algorithms -- $p_t(z_t)$ can be updated in a recursive manner, using rule~\eqref{eq:GPMWrule_v2} such distribution is re-computed at each round after observing $z_t$ (this requires storing the previous $\ucb_\tau$ functions, or re-computing them using \eqref{eq:posterior_mean_and_variance} at each round).\footnote{We note that strategy~\eqref{eq:GPMWrule_v2} can be implemented recursively in case the contexts' set $\mathcal{Z}$ is known and finite.} Such strategy exploits the stochastic assumption on the contexts and reduces the contextual game to a set of auxiliary games, one for each context.
This idea was recently used also by \cite{balseiro2019, neu2020} in the finite and linear contextual bandit setting, respectively, while we specialize it to repeated games coupled with our RKHS assumptions.

The next theorem provides a pseudo-regret bound for \textsc{c.GP-MW} when using strategy~\eqref{eq:GPMWrule_v2}, i.e., we bound the quantity $\mathbb{E} R_c^i(T,\pi)$ (expectation with respect to the contexts' sequence and the randomization of \textsc{c.GP-MW}), where $R_c^i(T,\pi)$ is the regret with respect to a generic policy $\pi\in \Pi$. Note that the pseudo-regret is smaller than the expected contextual regret $\mathbb{E} R_c^i(T)$ which, however, is proven to grow linearly with $T$ when $\mathcal{Z}$ is sufficiently large \cite{balseiro2019}. Nevertheless, \cite[Theorem 22]{balseiro2019} shows that $|\mathbb{E} R_c^i(T,\pi) - \mathbb{E} R_c^i(T)|$ can be bounded assuming each context occurs sufficiently often.\looseness=-1  
{\setlength\belowdisplayskip{0pt}
\setlength\abovedisplayskip{0pt}
\begin{theorem}\label{thm:thm2}
Fix $\delta \in (0,1)$ and assume $\| r^i\|_{k}\leq B$ and $z_t \sim \zeta$ for all $t$. Moreover, assume the opponents cannot observe the current context $z_t$. If player $i$ plays according to \textsc{c.GP-MW} using strategy~\eqref{eq:GPMWrule_v2} with $\lambda \geq 1$, $\beta_t = B + \sigma \lambda^{-1/2}\sqrt{2(\gamma_{t-1} + \log(1/\delta))}$ and $\eta_t  =  \sqrt{(8\log K)/T}$, then with probability at least $1- \delta$, 
\begin{equation*}
\sup_{\pi \in \Pi^i} \mathbb{E} \Big[ \sum_{t=1}^T r \big( \pi(z_t), a^{-i}_t, z_t\big) -   \sum_{t=1}^T  r\big(a_t, a^{-i}_t, z_t\big) \Big] \leq  \sqrt{0.5 T \log K}  + 4 \beta_T \sqrt{\gamma_T \lambda T} \,,
\end{equation*}
where expectation is with respect to both the contexts' sequence and the randomization of \textsc{c.GP-MW}.\looseness=-2
\end{theorem}}
The above guarantee significantly improves upon the ones obtained in the previous sections, as it does not depend on the context space $\mathcal{Z}$, and matches the regret of \textsc{GP-MW} in non-contextual games. It should be compared with the $\mathcal{O}(\sqrt{T K \log K})$ guarantee of \cite{balseiro2019} which assumes rewards for the non-revealed contexts are also observed, and the bandit $\mathcal{O}(\sqrt{cT K \log K})$ pseudo-regret of \cite{neu2020} assuming a linear dependence between contexts and rewards and a \emph{known} contexts distribution.
Exploiting our game assumptions, \textsc{c.GP-MW}'s performance decreases only logarithmically with $K$, relies on a more realistic feedback model than \cite{balseiro2019} and can deal with more complex rewards structures than \cite{neu2020}.\looseness=-1

\vspace{-0.5em}
\section{Game Equilibria and Efficiency}\label{sec:equilibria_and_efficiency}
\vspace{-0.5em}
In this section, we introduce new notions of equilibria and efficiency for contextual games. We recover game-theoretic learning results \cite{hart2000,roughgarden2015} showing that equilibria end efficiency (as defined below) can be approached when players minimize their contextual regret.\looseness=-1

\vspace{-0.5em}
\subsection{Contextual Coarse Correlated Equilibria}\label{subsec:cce}
\vspace{-0.5em}
A typical solution concept of multi-player \emph{static} games is the notion of Coarse Correlated Equilibria (CCEs) (see, e.g., \cite[Section 3.1]{roughgarden2015}). CCEs include Nash equilibria and have received increased attention because of their amenability to learning: a fundamental result from \cite{hart2000} shows that CCEs can be approached by decentralized no-regret dynamics, i.e., when each player uses a no-regret algorithm. These results, however, are not applicable to contextual games, where suitable notions of equilibria should capture the fact that players can observe the current context before playing. To cope with this, we define a notion of CCEs for contextual games, denoted as {\em contextual CCE (c-CCE)}.\looseness=-1
{\setlength\belowdisplayskip{5pt}
\setlength\abovedisplayskip{5pt}
\begin{definition}[Contextual CCE]\label{def:c-CCE} Consider a contextual game described by contexts $z_1, \ldots, z_T$. Let $\Pi^i$ be the set of all policies $\pi:  \mathcal{Z} \rightarrow \mathcal{A}^i$ for player $i$, and $\bm{\mathcal{A}}$ be the joint space of actions $\mathbf{a} = (a^i, a^{-i})$.
A \emph{contextual coarse-correlated equilibrium (c-CCE)} is a policy $\rho : \mathcal{Z} \rightarrow \Delta^{|\bm{\mathcal{A}}|}$ such that:\looseness=-1
\begin{equation}\label{eq:c-CCE}
    \frac{1}{T} \sum_{t=1}^T \: \mathop{\mathbb{E}}_{  \mathbf{a}  \sim \rho(z_t)} r^i\big(\mathbf{a}, z_t\big) \geq    \frac{1}{T} \sum_{t=1}^T  \:  \mathop{\mathbb{E}}_{  \mathbf{a}  \sim \rho(z_t)}  r^i \big(\pi(z_t), a^{-i},z_t \big) \quad \forall \pi \in  \Pi^i , \quad  \forall i = 1, \ldots, N  \,.
    \end{equation} 
    \vspace{-1em}
\end{definition}}
As opposed to CCEs (which are elements of $\Delta^{|\bm{\mathcal{A}}|}$), a c-CCE is a \emph{policy} $\rho : \mathcal{Z} \rightarrow \Delta^{|\bm{\mathcal{A}}|}$ from which no player has incentive to deviate looking at the time-averaged expected reward. 
In other words, suppose there is a trusted device that, for any context $z_t$, samples a joint action from $\rho(z_t)$, where $\rho$ is a c-CCE. Then, in expectation, each player is better off complying with such device, instead of using any other $\pi : \mathcal{Z} \rightarrow \mathcal{A}^i$. We say that $\rho$ is a $\epsilon$-c-CCE if inequality~\eqref{eq:c-CCE} is satisfied up to a $\epsilon\in \mathbb{R}_+$ accuracy. Finally, we remark that c-CCEs reduce to CCEs in case $z_t = z_0$ for all $t$.\looseness=-2

\vspace{-0.5em}
\paragraph{Example (c-CCEs in traffic routing)} In traffic routing applications, the trusted device can be a routing system (e.g., a maps server) which, given current weather, traffic conditions, or other contextual information, decides on a route for each user. If such routes are sampled according to a c-CCE, then each user is better-off complying with such device to ensure a minimum expected travel time.\looseness=-1

The next proposition shows that, similarly to CCEs in static games, c-CCEs can be approached whenever players minimize their contextual regrets. Hence, it provides a fully decentralized and efficient scheme for computing $\epsilon$-c-CCEs. To do so, we define the notion of \emph{empirical~policy} at round $T$ as follows. After $T$ game rounds, let $\mathcal{Z}_T$ be the set of all the distinct observed contexts. Then, the empirical~policy at round $T$ is the policy $\rho_T : \mathcal{Z} \rightarrow \Delta^{|\bm{\mathcal{A}}|}$ such that, for each $z \in \mathcal{Z}_T$, $\rho_T(z)$ is the empirical distribution of played actions when context $z$ was revealed, while for the unseen contexts $z \in \mathcal{Z} \setminus \mathcal{Z}_T$, $\rho_T(z)$ is an arbitrary (e.g., uniform) distribution.
\begin{proposition}[Finite-time approximation of c-CCEs]\label{prop:convergence_c-CCE} After $T$ game rounds, let $R^i_c(T)$'s denote the players' contextual regrets and $\rho_T$ be the empirical policy at round $T$. Then, $\rho_T$ is a $\epsilon$-c-CCE of the played contextual game with
$ \epsilon \leq  \max_{i \in \{1, \ldots, N\} } R^i_c(T)/T  $.
\end{proposition}
Proposition~\ref{prop:convergence_c-CCE} implies that, as $T \rightarrow \infty$, if players use vanishing contextual regret algorithms (such as the ones discussed in Section~\ref{sec:algorithms}), then the empirical policy $\rho_T$ converges to a c-CCE of the contextual game. 
When contexts are stochastic, i.e., $z_t \sim \zeta$ for all $t$, an alternative notion of c-CCE can be defined by considering the \emph{expected} context realization. We treat this case in Appendix
~\ref{app:stochastic_cCCE} and prove similar finite-time and asymptotic convergence results using standard concentration arguments.\looseness=-1

\setlength\belowdisplayskip{5pt}
\setlength\abovedisplayskip{5pt}
\vspace{-0.5em}
\subsection{Approximate Efficiency}
\vspace{-0.5em}
The efficiency of an outcome $\mathbf{a} \in \bm{\mathcal{A}}$ in a non-contextual game, i.e., for a fixed context $z_0$, can be quantified as the distance between the \emph{social welfare}  $\Gamma(\mathbf{a}, z_0) := \sum_{i=1}^N r^i(\mathbf{a}, z_0)$ (where $:=$ is sometimes replaced by $\geq$ if the reward of the game authority is also considered), and the optimal welfare 
$\max_{\mathbf{a}}\Gamma(\mathbf{a}, z_0)$. Optimal welfare is typically not achieved as players are self-interested agents aiming at maximizing their individual rewards, instead of $\Gamma$.
Nevertheless, main results of \cite{blum2008, roughgarden2015} show that such efficiency loss can be bounded whenever players minimize their regrets and the game is \emph{$(\lambda, \mu)$-smooth}, i.e., if for any pair of outcomes $\mathbf{a}_1= (a_1^1, \ldots, a_1^N)$ and $\mathbf{a}_2 = (a_2^1, \ldots, a_2^N)$, it satisfies:
{\setlength\belowdisplayshortskip{0pt}
\setlength\abovedisplayskip{-5pt}
\setlength\abovedisplayshortskip{-7pt}
\begin{equation}\label{eq:smoothness}
 \sum_{i=1}^N   r^i( a_2^i,  a_1^{-i}, z_0) \geq \lambda \cdot \Gamma( \mathbf{a}_2, z_0) \: -   \mu \cdot \Gamma(\mathbf{a}_1, z_0) \,.
\end{equation}
}
Examples of smooth games are routing games with polynomial delay functions, several classes of auctions, submodular welfare games, and many more (see, e.g., \cite{syrgkanis2013, roughgarden2015, sessa2019bounding} and references therein).

In contextual games, on the other hand, a different context $z_t$ describes the game at each round, and hence the social welfare $\Gamma(\mathbf{a}, z_t)$ of an outcome $\mathbf{a} \in \bm{\mathcal{A}}$ depends on the specific context realization. The efficiency of a contextual game can therefore be quantified by the optimal \emph{contextual} welfare:
\begin{definition}[Optimal contextual welfare]
\begin{equation}\label{eq:social_optimum}
\mathrm{OPT} = \max_{ \pi^1 \in \Pi^1,  \ldots , \pi^N \in   \Pi^N } \:  \frac{1}{T }\sum_{t=1}^T \Gamma \big( \pi^1(z_t), \ldots, \pi^N(z_t), z_t\big) \,.
\end{equation}
\vspace{-1.1em}
\end{definition}
Equation~\eqref{eq:social_optimum} generalizes the optimal welfare for non-contextual games, and sets the stronger benchmark of finding the \emph{policies} (instead of static actions) mapping contexts to actions which maximize the time-averaged social welfare. In routing games, for instance, it corresponds to finding the best routes for the agents as a function of current traffic conditions. As shown in our experiments (Section~\ref{sec:experiments}), such policies can significantly reduce the network's congestion compared to finding the best static routes.
The next proposition generalizes the well-known results of \cite{roughgarden2015}, showing that in a smooth contextual game such optimal welfare can be approached when players minimize their contextual regret. First, we note that a contextual game can satisfy the smoothness condition \eqref{eq:smoothness} for different constants $\lambda, \mu$, depending on the context $z_t$, and hence we use the notation $\lambda(z_t), \mu(z_t)$ to highlight their dependence.\looseness=-2
{
\begin{proposition}[Convergence to approximate efficiency]\label{prop:efficiency}
Let $R_c^i(T)$'s be the players' contextual regrets, and assume the game is $(\lambda(z_t), \mu(z_t))$-smooth at \emph{each} time $t$.
Then,
\begin{equation*}
 \frac{1}{T }\sum_{t=1}^T \Gamma \big( a^1_t,\ldots, a^N_t, z_t\big) \geq \frac{\bar{\lambda}}{1 + \bar{\mu}} \: \mathrm{OPT}  \: -  \frac{1}{1 + \bar{\mu}}  \sum_{i=1}^N \frac{R^i_c(T)}{T}   \, ,
\end{equation*}
where $\bar{\lambda} = \max_{t\in \{1,\ldots, T\}} \lambda(z_t)$ and $\bar{\mu} = \min_{t\in \{1,\ldots, T\}} \mu(z_t)$.
\vspace{-0.5em}
\end{proposition}}
The approximation factor $\bar{\lambda}/(1 + \bar{\mu})$ (also known as Price of Total Anarchy \cite{blum2008}) depends on the constants $\bar{\lambda}$ and $\bar{\mu}$, which in our case represent the `worst-case' smoothness of the game. We remark however that game smoothness is not necessarily context-dependent, e.g., routing games (such as the one considered in the next section) are smooth regardless of the network's size and capacities \cite{roughgarden2015}.

\setlength{\textfloatsep}{10pt}
\begin{figure}[t!]
    \centering
    \hspace{-1em}   
    \includegraphics[width=0.35 \textwidth]{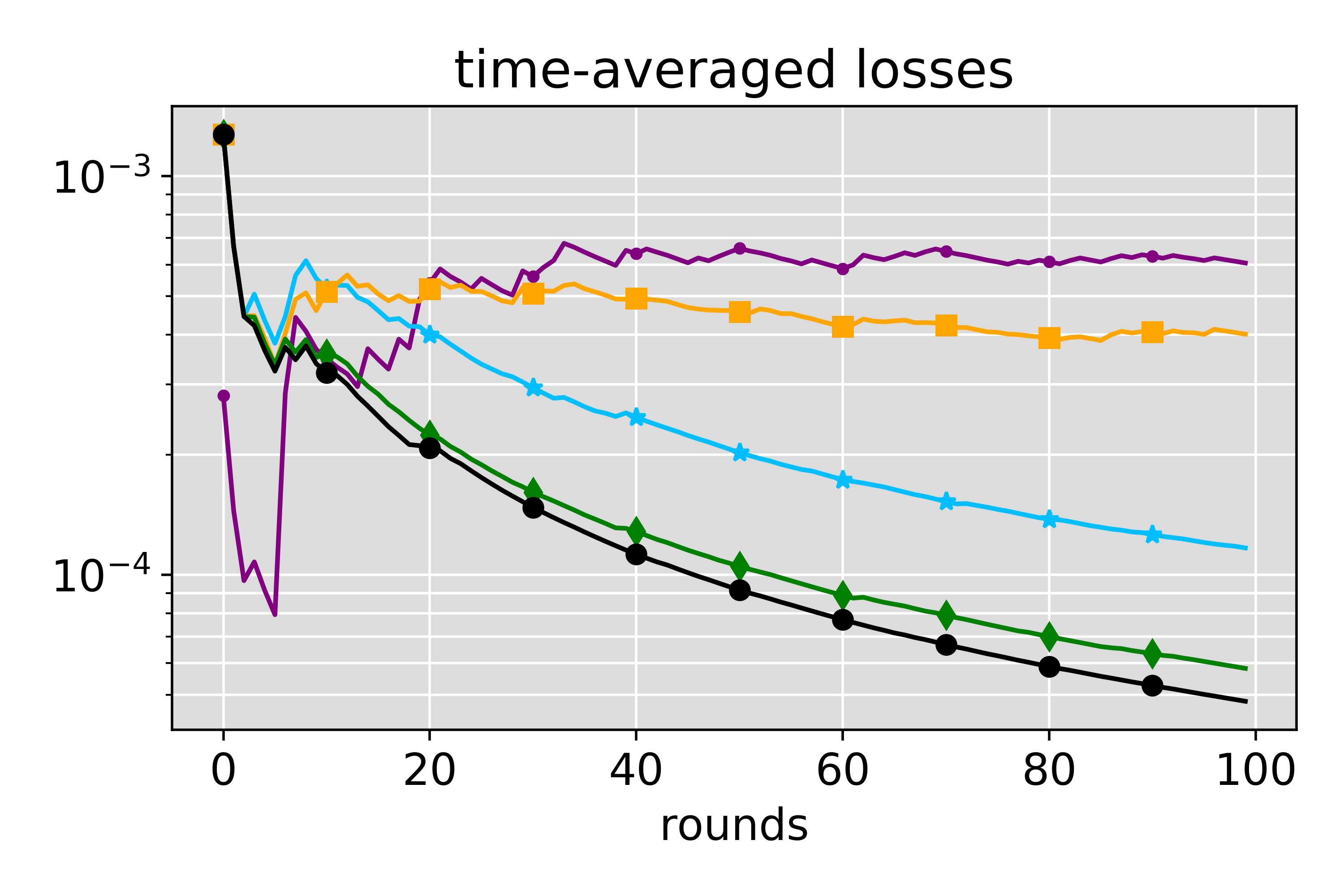}
    \hspace{2em}
     \includegraphics[width=0.35 \textwidth]{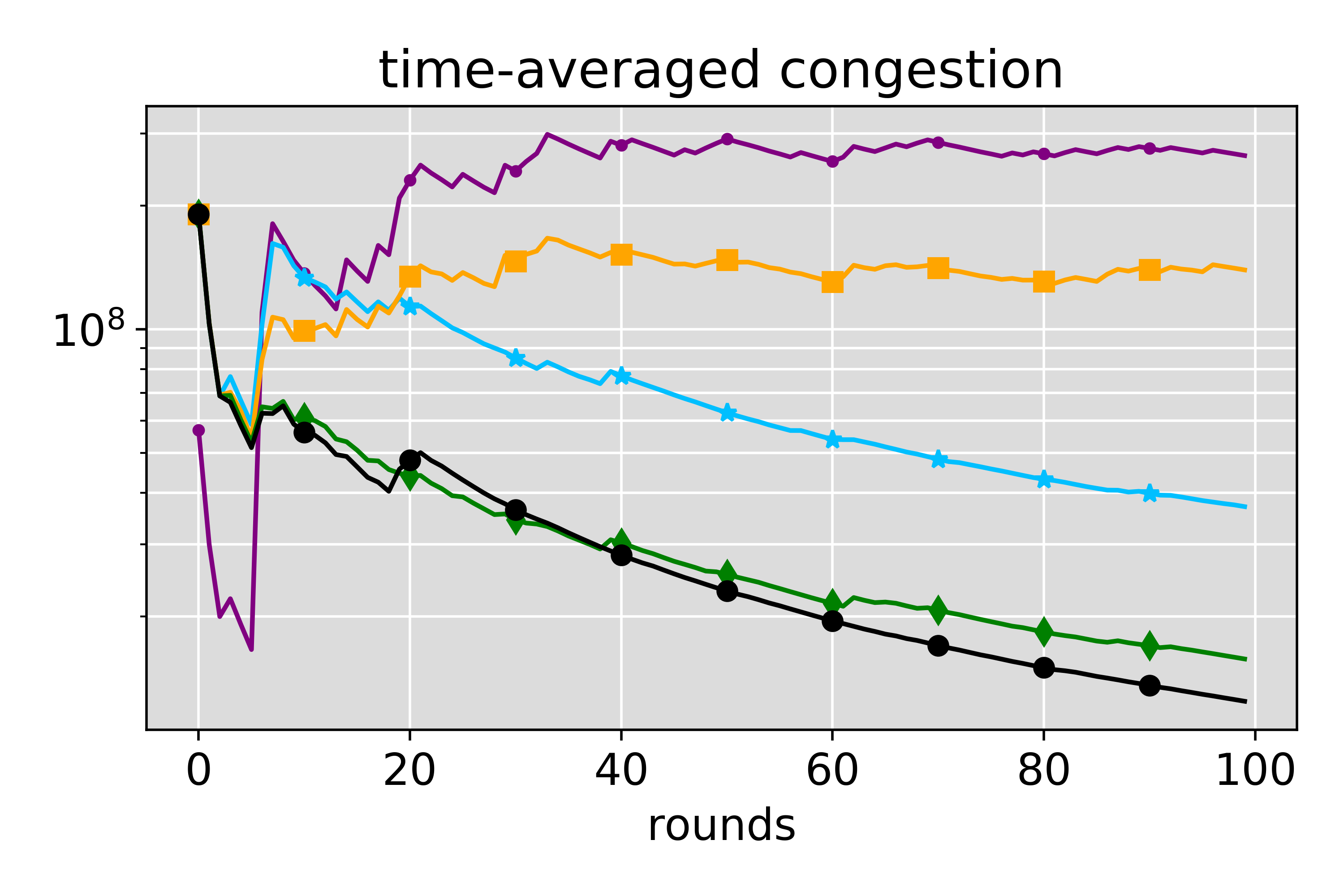}
          \vspace{-0.5em}
          
  \includegraphics[width=0.9 \textwidth]{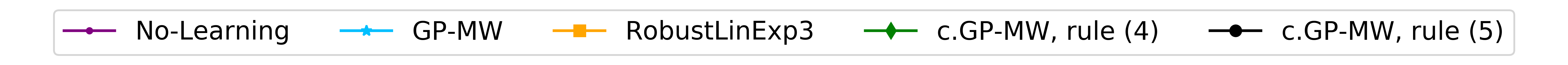}
  \vspace{-0.3em}
     \caption{Time-averaged losses (Left) and network congestion (Right), when agents use different routing strategies (average over 5 runs). \textsc{c.GP-MW} leads to reduced losses and congestion compared to the other baselines.\looseness=-2}
    \label{fig:routing_congestions}
\end{figure}{}

\vspace{-0.5em}
\section{Experiments - Contextual Traffic Routing Game}\label{sec:experiments}
\vspace{-0.5em}

We consider a \emph{contextual} routing game on the traffic network of Sioux-Falls, a directed graph with $24$ nodes and $76$ edges, with the same game setup of \cite{sessa2019noregret} (network data and congestion model are taken from \cite{leblanc1975}, while Appendix~\ref{app:routing} gives a complete description of our experimental setup). There are $N = 528$ agents in the network. Each agent wants to send $d_i$ units from a given origin to a given destination node in minimum time, and can choose among $K=5$ routes at each round. The traveltime of an agent depends on the routes chosen by the other agents (if all choose the same route the network becomes highly congested) as well as the network's capacity at round $t$. We let $x_t
^i \in \mathbb{R}^{76}$ represent the route chosen by agent $i$ at round $t$, where $x_t^i[e] = d_i$ if edge $e$ belongs to such route, and $x_t^i[e] = 0$ otherwise. Moreover, we let context $z_t\in \mathbb{R}_+^{76}$ represent the capacity of the network's edges at round $t$ (capacities are i.i.d. samples from a fixed distribution, see Appendix~\ref{app:routing}). Then, agents' rewards can be written as $r^i(x_t^i, x_t^{-i},z_t) = - \sum_{e =1}^{76} x_t^i[e] \cdot t_e(x_t^i + x_t^{-i}, z_t[e])$, where $x_t
^{-i} = \sum_{j \neq i} x_t^j$ and $t_e(\cdot)$'s are the edges' traveltime functions, which are unknown to the agents. Note that, strictly speaking, agent~$i$'s reward depends only on the entries $x_t^i[e], x_t^{-i}[e], z_t[e]$ for $e \in E^i$, where $E^i$ is the subset of edges that agent $i$ could potentially traverse. According to our model, we assume agent~$i$ observes context $\{z_t[e], e\in E^i \}$ and, at end of each round, the edges' occupancies $\{x_t^{-i}[e], e\in E^i \}$. 

We let each agent select routes according to \textsc{c.GP-MW} (using rule \eqref{eq:GPMWrule_lipschitz}  or \eqref{eq:GPMWrule_v2}) and compare its performance with the following baselines: 1) No-Learning, i.e., agents select the shortest free-flow routes at each round, 2) \textsc{GP-MW}~\cite{sessa2019noregret} which neglects the observed contexts, and 3) \textsc{RobustLinExp3}~\cite{neu2020} for contextual linear bandits, which is robust to model misspecification but does not exploit the correlation in the game (also, it requires knowing the contexts' distribution).
To run \textsc{c.GP-MW} we use the composite kernel $k(x_t^i, x_t^{-i}, z_t) = k_1(x_t^i) * k_2((x_t^i + x_t^{-i})/z_t)$, while for \textsc{GP-MW} the kernel $k(x_t^i, x_t^{-i}, z_t) = k_1(x_t^i) * k_2(x_t^i + x_t^{-i})$, where $k_1,k_2$ are linear and polynomial kernels respectively. We set $\eta_t$ according to Theorems~\ref{thm:thm1} and \ref{thm:thm2}, and $\beta_t = 2.0$ (theoretical values for $\beta_t$ are found to be overly conservative \cite{srinivas2009gaussian,sessa2019noregret}). For rule~\eqref{eq:GPMWrule_lipschitz} we set $\epsilon= 30|E^i|$. 
Figure~\ref{fig:routing_congestions} shows the time-averaged losses (i.e., traveltimes scaled in $[0,1]$ and averaged over all the agents), which are inversely proportional to the game welfare, and the resulting network's congestion (computed as in Appendix~\ref{app:routing}). We observe, in line to what discussed in Section~\ref{sec:equilibria_and_efficiency}, that minimizing individual regrets the agents increase the game welfare (this is expected as routing games are smooth \cite{roughgarden2007}). Moreover, when using \textsc{c.GP-MW} agents exploit the observed contexts and correlations, and achieve significantly more efficient outcomes and lower congestion levels compared to the other baselines. 
We also observe strategy~\eqref{eq:GPMWrule_v2} outperforms strategy~\eqref{eq:GPMWrule_lipschitz} in our experiments. This can be explained by the contexts being stochastic and, also, since each agent $i$ is only influenced by the coordinates $z_t[e]$ of the relevant edges $e \in E^i$.\looseness=-1

\vspace{-0.5em}
\section{Conclusions}
\vspace{-0.5em}
We have introduced the class of contextual games, a type of repeated games described by contextual information at each round. Using kernel-based regularity assumptions, we modeled the correlation between different contexts and game outcomes, and proposed novel online algorithms that exploit such correlations to minimize the players' contextual regret. We defined the new notions of contextual Coarse Correlated Equilibria and optimal contextual welfare and showed that these can be approached when players have vanishing contextual regret. The obtained results were validated in a traffic routing experiment, where our algorithms led to reduced travel times and more efficient outcomes compared to other baselines that do not exploit the observed contexts or the correlation present in the game.\looseness=-1

\section*{Broader Impact}
As systems using machine learning get deployed more and more widely, 
these systems increasingly interact with each other.  Examples range from road traffic over auctions and financial markets, to robotic systems. Understanding these interactions and their effects for individual participants and the reliability of the overall system becomes ever more important. We believe our work contributes positively to this challenge by studying principled algorithms that are efficient, while converging to suitable, and often efficient, equilibria.

\subsubsection*{Acknowledgments}
This work was gratefully supported by the Swiss National Science Foundation, under the grant SNSF $200021$\textunderscore$172781$, by the European Union's ERC grant $815943$, and ETH Z\"urich Postdoctoral Fellowship 19-2 FEL-47.

\bibliographystyle{plain}
\bibliography{main.bib}

\newpage
\appendix

{\centering
    {\huge \bf Supplementary Material}
    
    {\Large \bf Contextual Games: Multi-Agent Learning with Side Information \\ [2mm] {\normalsize \bf {Pier Giuseppe Sessa, Ilija Bogunovic, Andreas Krause, Maryam Kamgarpour (NeurIPS 2020)} \par }  
}}

\setlength\belowdisplayskip{8pt}
\setlength\abovedisplayskip{8pt}
\vspace{1em}

\section{Supplementary Material for Section~\ref{sec:algorithms}}

The theoretical guarantees obtained in Section~\ref{sec:algorithms} rely on the following two main lemmas. The first lemma is from \cite{abbasi2013online} and shows that given the previously observed rewards, contexts, and players' actions, the reward function of player~$i$ belongs (with high probability) to the interval $[\mu_t(\cdot, \cdot) \pm \beta_t \sigma_t(\cdot,\cdot)]$, for a carefully chosen confidence parameter $\beta_t \geq 0$.

\begin{lemma} \label{lemma:conf_lemma}
Let $r \in \mathcal{H}_k$ such that $\|r\|_{k} \leq B$ and consider the kernel-ridge regression mean and standard deviation estimates $\mu_t(\cdot)$ and $ \sigma_t(\cdot)$, with regularization constant $\lambda > 0$. Then for any $\delta \in (0,1)$, with probability at least $1 - \delta$, the following holds simultaneously over all $x \in \mathcal{D}$ and $t\geq 1$:
    \begin{equation*}
        |\mu_t(x) - r(x)| \leq \beta_t \sigma_t(x),
    \end{equation*}
where $\beta_t = B\lambda^{-1/2} + \sigma \lambda^{-1} \sqrt{2 \log{(\tfrac{1}{\delta})} + \log(\det(I_t + K_t/\lambda))}$.
\end{lemma}
Therefore, according to Lemma~\ref{lemma:conf_lemma}, the function $\ucb_t$ defined in \eqref{eq:ucb} represents a valid upper confidence bound for the rewards obtained by player $i$.

The second main lemma concerns the properties of the Multiplicative Weights (MW) update method~\cite{littlestone1994}, which is used as a subroutine in our algorithms to compute the action distribution $p_t(z_t)$ (Line 3 of Algorithm~\ref{alg:algorithm1}) at each round. Its proof follows from standard online learning arguments equivalently to, e.g., \cite[Proposition 1]{Mourtada2019}.

\begin{lemma}\label{lem:MW_rule}
Consider a sequence of functions $g_1(\cdot), \ldots, g_T(\cdot) \in [0,1]^K$ and let $p_t$'s be the distributions computed using the MW rule: 
\begin{equation}\label{eq:MW_rule_explicited}
p_t[a] \propto \exp\left( \eta_{t}  \cdot \sum_{\tau=1}^{t-1}  g_\tau(a) \right)   \qquad a=1, \ldots, K \, ,
\end{equation}
where $p_1$ is initialized as the uniform distribution.
Then, provided that $\{ \eta_t \}_{t=1}^T$ is a decreasing sequence, for any action $a^\star \in \{1,\ldots,K\}$:
\begin{equation*}
\sum_{t=1}^T g_t(a^\star) - \sum_{t=1}^T \sum_{a =1}^K p_t[a] \cdot g_t(a) \leq \frac{\log K}{\eta_T} + \frac{\sum_{t=1}^T \eta_t}{8} \,.
\end{equation*}
\end{lemma}

\subsection{The case of a finite (and small) number of contexts}\label{app:simple_case}
In this section we consider the simple case of a finite (and small cardinality) set of contexts $\mathcal{Z}$. In such a case, a high-probability regret bound of $\mathcal{O}(\sqrt{T|\mathcal{Z}|\log K} + \gamma_T \sqrt{T})$ can be achieved when \textsc{c.GP-MW} is run with the following strategy:
\begin{equation}\label{eq:cGPMW_rule}
p_t(z_t)[a] \propto \exp\left( \eta_t \cdot \sum_{\tau=1}^{t-1} \ucb_\tau(a, a^{-i}_\tau, z_\tau)\cdot \mathds{1}\{z_\tau = z_t \} \right)   \qquad a=1, \ldots, K \, .
\end{equation}
That is, $p(z_t)$ is computed using the sequence of previously computed upper confidence bound functions for the game rounds in which the specific context $z_t$ was revealed. 

\begin{theorem}\label{thm:thm0}
Fix $\delta \in (0,1)$ and assume $\| r^i\|_{k^i}\leq B$. If player $i$ plays according to \textsc{c.GP-MW} using strategy~\eqref{eq:cGPMW_rule} with $\lambda \geq 1$, $\beta_t = B + \sigma\lambda^{-1/2} \sqrt{2(\gamma_{t-1} + \log(2/\delta))}$, and $\eta_t = 2\sqrt{\log K/\sum_{\tau=1}^{t}\mathds{1}\{z_\tau = z_t\}}$, then with probability at least $1- \delta$,
\begin{equation*}
R_c^i(T) \leq  \sqrt{T  |\mathcal{Z}| \log K} + \sqrt{0.5 T \log(2/ \delta)}  + 4 \beta_T \sqrt{\gamma_T \lambda T} \, .
\end{equation*}
\end{theorem}
\begin{proof}
Let $\pi^\star = \arg\max_{\pi \in \Pi^i} \sum_{t=1}^T    r \big( \pi(z_t), a^{-i}_t, z_t\big)$. Our goal is to bound $R_c^i(T) = \sum_{t=1}^T    r \big( \pi^\star(z_t), a^{-i}_t, z_t\big) -     r \big(a^i_t, a^{-i}_t, z_t\big)$
with high probability. 

By conditioning on the event of the confidence lemma (Lemma~\ref{lemma:conf_lemma}) holding true, we can state that, with probability at least $1-\delta/2 $,
\begin{align}
 R_c^i(T)& = \sum_{t=1}^T    r \big( \pi^\star(z_t), a^{-i}_t, z_t\big) -     r \big(a^i_t, a^{-i}_t, z_t\big)   \nonumber \nonumber \\
& \leq  \sum_{t=1}^T  \Big(  \ucb_t \big( \pi^\star(z_t), a^{-i}_t, z_t\big) -     \ucb_t \big(a^i_t, a^{-i}_t, z_t\big) \Big) + \sum_{t=1}^T  2 \beta_t \sigma_t \big( a^i_t, a^{-i}_t, z_t \big)  \nonumber \\
& \leq  \underbrace{ \sum_{t=1}^T    \Big(  \ucb_t \big( \pi^\star(z_t), a^{-i}_t, z_t\big) -     \ucb_t \big(a^i_t, a^{-i}_t, z_t\big)  \Big) }_{\hat{R^i_c}(T)} + 4 \beta_T \sqrt{\gamma_T \lambda T} \, .\label{eq:bound_regret_with_ucb}
\end{align}
The first inequality follows by the definition of $\ucb_t(\cdot)$ (see \eqref{eq:ucb}), the specific choice of the confidence level $\beta_t$, and Lemma~\ref{lemma:conf_lemma}. The last inequality follows by \cite[Lemma 5.4]{srinivas2009gaussian}. The rest of the proof proceeds to show that, with probability at least $1-\delta/2$,  
\begin{equation}\label{eq:regret_final}
    \hat{R^i_c}(T) \leq  \sqrt{T  |\mathcal{Z}| \log K} + \sqrt{0.5 T \log(2/ \delta)} \, .
\end{equation}
The theorem statement then follows by a standard union bound argument.

First, by straightforward application of the Hoeffding–Azuma inequality (e.g., \cite[Lemma A.7]{cesa-bianchi_prediction_2006}), it follows that with probability at least $1-\delta/2$,
\begin{equation}\label{eq:azuma}
 \sum_{t=1}^T \big|  \underbrace{ \ucb_t\big(a^i_t, a^{-i}_t, z_t\big)  - \sum_{a \in \mathcal{A}^i}p_t(z_t)[a] \cdot \ucb_t\big(a, a^{-i}_t, z_t\big)}_{X_t} \big| \leq \sqrt{0.5 T \log(2/ \delta)} \, ,
\end{equation}
since the variables $X_t$'s form a martingale difference sequence, being $\sum_{a \in \mathcal{A}^i} p_t(z_t)[a] \cdot \ucb_t\big(a, a^{-i}_t, z_t\big)$ the expected value of $\ucb_t\big(a^i_t, a^{-i}_t, z_t\big)$ conditioned on the history $\{a_\tau^i, a_\tau^{-i}, z_\tau, \epsilon_\tau \}_{\tau=1}^{t-1}$ and on context $z_t$.
Then, using \eqref{eq:azuma}, $\hat{R^i_c}(T)$ can be bounded, with probability $1-\delta/2$, as
\begin{align}
    \hat{R^i_c}(T) & \leq \sum_{t= 1}^T    \ucb_t \big( \pi^\star(z_t), a^{-i}_t, z_t\big) - \sum_{a \in \mathcal{A}^i}p_t(z_t)[a] \cdot \ucb_t\big(a, a^{-i}_t, z_t\big)  \: + \sqrt{0.5 T \log(2/ \delta)} \nonumber
    \\ 
    & = \sum_{z \in \mathcal{Z} } \sum_{t: z_t = z}    \ucb_t \big( \pi^\star(z_c), a^{-i}_t, z_c\big) -     \sum_{a \in \mathcal{A}^i}p_t(z_c)[a] \cdot \ucb_t\big(a, a^{-i}_t, z_c\big) \: + \sqrt{0.5 T \log(2/ \delta)} \,. \label{eq:bound_R_hat}
\end{align}
At this point, we can use the properties of the MW rule
used to compute the distribution $p_t(z) \in \Delta^K$.
Note that, for each context $z \in \mathcal{Z}$, the distribution $p_t(z)$ computed by \textsc{c.GP-MW} precisely follows the MW rule~\eqref{eq:MW_rule_explicited} with the sequence of functions $\{\ucb_\tau(\cdot, a_\tau^{-i}, z)\}_{\tau: z_\tau = z}$ and the sequence of learning rates $\{\eta_\tau\}_{\tau=1}^{T_z} = \{2\sqrt{\log K / \tau} \}_{\tau=1}^{T_z}$, where $T_z = \sum_{\tau=1}^{t}\mathds{1}\{z_\tau = z_t\}$ is the number of times context $z$ was revealed. Hence, we can apply Lemma~\ref{lem:MW_rule} for each context $z$ and obtain: 
\begin{align}
    &  \sum_{t: z_t = z} \ucb_t \big( \pi^\star(z_c), a^{-i}_t, z_c\big) -     \sum_{a \in \mathcal{A}^i}p_t(z_c)[a] \cdot \ucb_t\big(a, a^{-i}_t, z_c\big) \leq \nonumber \\
    & \hspace{10em} \leq 0.5 \sqrt{T_z \log K} + \frac{2\sqrt{\log K}}{8} \sum_{\tau= 1}^{T_z} \frac{1}{\sqrt{\tau}} \nonumber \\ 
    &\hspace{10em} \leq 0.5 \sqrt{T_ z\log K} + \frac{2\sqrt{\log K}}{8} 2 \sqrt{T_z}  =  \sqrt{T_z\log K} \,. \label{eq:MWbound_per_context}
\end{align}
Equation~\eqref{eq:MWbound_per_context}, together with the bound \eqref{eq:bound_R_hat} leads to
\begin{align}
    \hat{R^i_c}(T) & \leq \sum_{z \in \mathcal{Z} }\sqrt{ T_z  \log K}  \: + \sqrt{0.5 T \log(2/ \delta)} \nonumber \\
    & \leq  \sqrt{ T  |\mathcal{Z}| \log K} \: + \sqrt{0.5 T \log(2/ \delta)}  \, , \nonumber
\end{align}
where in the last inequality we have used Cauchy–Schwarz inequality and $\sum_{z \in \mathcal{Z}} T_z = T$. Hence, we finally proved \eqref{eq:regret_final}.
Therefore, with probability at least $1-\delta/2 -\delta/2 = 1-\delta$ we obtain the final regret bound combining \eqref{eq:bound_regret_with_ucb} and \eqref{eq:regret_final}.
\end{proof}

\allowdisplaybreaks

\subsection{Proof of Theorem~\ref{thm:thm1}}\label{app:proof_thm1}
\begin{proof}
Similarly to Appendix~\ref{app:simple_case}, we let $\pi^\star = \arg\max_{\pi \in \Pi^i} \sum_{t=1}^T    r \big( \pi(z_t), a^{-i}_t, z_t\big)$ and seek to bound $R_c^i(T) = \sum_{t=1}^T    r \big( \pi^\star(z_t), a^{-i}_t, z_t\big) -     r \big(a^i_t, a^{-i}_t, z_t\big)$
with high probability. Recall that Strategy~\ref{strategy} builds an $\epsilon$-net of the contexts space by creating new L1-balls in a greedy fashion. After $T$ game rounds, the set $\mathcal{C}$ contains the centers $z \in \mathcal{Z}$ of the balls created so far. Moreover, at each round $t$, the variable $z_t'$ indicates the ball that context $z_t$ has been associated to. According to this notation, player~$i$'s regret can be rewritten as 
\begin{align*}
 & R_c^i(T)  = \sum_{z \in \mathcal{C}} \: \sum_{t: z_t' = z}  r \big( \pi^\star(z_t), a^{-i}_t, z_t\big) -     r \big(a^i_t, a^{-i}_t, z_t\big) \\ 
& =  \underbrace{\sum_{z \in \mathcal{C}}\: \sum_{t: z_t' = z}  r \big( \pi^\star(z_t), a^{-i}_t, z_t\big) - r \big(\pi^\star(z), a^{-i}_t, z_t\big)}_{R_L(T)}  + \underbrace{\sum_{z \in \mathcal{C}}\: \sum_{t: z_t' = z} r \big(\pi^\star(z), a^{-i}_t, z_t\big) -   r \big(a^i_t, a^{-i}_t, z_t\big)}_{R_C(T)} \,,
\end{align*}
where in the last equality we have added and subtracted the term $\sum_{z \in \mathcal{C}}\: \sum_{t: z_t' = z} r \big(\pi^\star(z), a^{-i}_t, z_t\big)$.

The regret term $R_L(T)$ can be bounded using $L_r$-Lipschizness of $r(\cdot)$ in its first argument and 
$L_p$-Lipschizness of the optimal policy: 
\begin{align*}
R_L(T) & \leq \sum_{z \in \mathcal{C}} \: \sum_{t: z_t' = z} L_r \|\pi^\star(z_t) - \pi^\star(z) \|_1 \leq \sum_{z \in \mathcal{C}} \: \sum_{t: z_t' = z} L_r L_p \|z_t - z \|_1  \\ 
& \leq L_r L_p T \epsilon  \, ,
\end{align*}
where in the last step we have used that, when $z_t'= z$, $z_t$ belongs to the ball centered at $z$.

We can now proceed in bounding $R_C(T)$. 
Note that we can apply the same proof steps of Appendix~\ref{app:simple_case} (namely, Equations \eqref{eq:bound_regret_with_ucb} and \eqref{eq:azuma}) to show that, with probability at least $1-\delta$, we have: 
\begin{align*}
    R_C(T) & \leq  \sum_{z \in \mathcal{C}} \: \sum_{t: z_t' = z}   \ucb_t \big( \pi^\star(z_t), a^{-i}_t, z_t\big) - \sum_{a \in \mathcal{A}^i}p_t(z_t)[a] \cdot \ucb_t\big(a, a^{-i}_t, z_t\big)  \\
    & \hspace{20em} + 4 \beta_T \sqrt{\gamma_T \lambda T} + \sqrt{0.5 T \log(2/ \delta)} \nonumber
     \,.
\end{align*}
where we have conditioned on the event of Lemma~\ref{lemma:conf_lemma} and applied the Hoeffding-Azuma inequality. At this point, we can use the properties of the MW rule used in Strategy~\ref{strategy} to compute $p_t(z_t)$ at each round. Note that Strategy~\ref{strategy} corresponds to maintaining a distribution for each $z\in \mathcal{C}$ and update it when context $z_t$ belongs to such ball. After $T$ rounds, for each $z \in \mathcal{Z}$ let $T_z= \sum_{t=1}^T \mathds{1} \{ z_t' =z\}$ be the number of times the revealed context belonged to the ball centered at $z$.  
Hence, for each $z\in \mathcal{C}$, a straightforward application of Lemma~\ref{lem:MW_rule} leads to: 
\begin{align}
    &  \sum_{t: z_t' = z}   \ucb_t \big( \pi^\star(z_t), a^{-i}_t, z_t\big) - \sum_{a \in \mathcal{A}^i}p_t(z_t)[a] \cdot \ucb_t\big(a, a^{-i}_t, z_t\big) \leq \nonumber \\
    & \hspace{10em} \leq 0.5 \sqrt{\log K T_z} + \frac{2\sqrt{\log K}}{8} \sum_{\tau= 1}^{T_z} \frac{1}{\sqrt{\tau}} \nonumber \\ 
    &\hspace{10em} \leq 0.5 \sqrt{\log K T_z} + \frac{2\sqrt{\log K}}{8} 2 \sqrt{T_z}  =  \sqrt{T_z \log K} \,. \nonumber
\end{align}
where we have used the same steps to obtain  Equation~\eqref{eq:MWbound_per_context} in Appendix~\ref{app:simple_case}. Hence, summing over all the contexts in $\mathcal{C}$ we obtain
\begin{align*}
    R_C(T) & \leq \sum_{z \in \mathcal{C}} \sqrt{ T_c  \log K} \quad + 4 \beta_T \sqrt{\gamma_T \lambda T} + \sqrt{0.5 T \log(1/ \delta_2)} \\
    & \leq  \sqrt{T |\mathcal{C}| \log K} \quad  + 4 \beta_T \sqrt{\gamma_T \lambda T} + \sqrt{0.5 T \log(1/ \delta_2)}\\
    & = \epsilon^{-c/2} \sqrt{T \log K} \quad  + 4 \beta_T \sqrt{\gamma_T \lambda T} + \sqrt{0.5 T \log(1/ \delta_2)} \,.
\end{align*}
In the last inequality we have used $|\mathcal{C}|\leq (1/\epsilon)^c$, because the contexts space $\mathcal{Z} \subseteq [0,1]^c$ can be covered by at most $(1/\epsilon)^c$ balls of radius $\epsilon$ such that the distance between their centers is at least $\epsilon$.

Therefore, combining the bounds for $R_L(T)$ and $R_C(T)$, the contextual regret of player $i$ is bounded, with probability at least $1-\delta$, by: 
\begin{align*}
R_c^i(T) & \leq  L_r L_p T \epsilon + \epsilon^{-c/2} \sqrt{ T  \log K} + 4 \beta_T \sqrt{\gamma_T \lambda T} + \sqrt{0.5 T \log(2/ \delta)} \\ 
& = (L_r L_p)^{\frac{c}{c + 2}} \:  T^\frac{c+1}{c+2} + (L_r L_p)^{\frac{c}{c + 2}} T^\frac{c+1}{c+2}  \sqrt{\log K}  + 4 \beta_T \sqrt{\gamma_T \lambda T} +  \sqrt{0.5 T \log(2/ \delta)} \, .
\end{align*}
where we have substituted the choice of $\epsilon = (L_r L_p)^{-\frac{2}{c + 2}} T^{-\frac{1}{c+2}}$.
\end{proof}

\subsection{Proof of Theorem~\ref{thm:thm2}}
\begin{proof}
We let $R_c^i (\pi, T) = \sum_{t=1}^T    r \big( \pi(z_t), a^{-i}_t, z_t\big) -     r \big(a^i_t, a^{-i}_t, z_t\big)$ be the regret of player $i$ with respect to a generic policy $\pi :\mathcal{Z} \rightarrow \mathcal{A}^i$. Our goal is to bound $R_c^i (\pi, T)$ in expectation, with respect to the random sequence of contexts and played actions. For ease of exposition, we use the notation $v_{\cdots T}$ to indicate the sequence of variables $v_1, \ldots, v_T$.  Moreover, we will explicitly consider an adaptive adversary that selects $a_t^{-i}$ as a function $a_t^{-i} = f(\mathcal{H}_{t-1})$ of the history $\mathcal{H}_{t-1}:= \{a^i_{\cdots t-1}, z_{\cdots t-1} \}$ but not of $z_t$.

First, note that the expected value of $R_c^i (\pi, T)$ is still a random variable which depends on the realization of the observation noise $\epsilon_t$'s. As it was done in proof of Theorems~\ref{thm:thm1} and Appendix~\ref{app:simple_case}, we can condition on the event of the confidence Lemma~\ref{lemma:conf_lemma}, and state that, with probability at least $1-\delta/2$, it can be bounded by

\begin{align}
 \mathbb{E}_{\substack{z_{\cdots T} \\ a^i_{\cdots T}}} \big[ R_c^i(\pi, T)\big ] & =    \mathbb{E}_{\substack{z_{\cdots T} \\ a^i_{\cdots T}}} \left[ \sum_{t=1}^T    r \big( \pi(z_t), a^{-i}_t, z_t\big) -     r \big(a^i_t, a^{-i}_t, z_t\big) \right]   \nonumber \nonumber \\
& \leq  \mathbb{E}_{\substack{z_{\cdots T} \\ a^i_{\cdots T}}} \left[ \sum_{t=1}^T    \ucb_t \big( \pi(z_t), a^{-i}_t, z_t\big) -     \ucb_t \big(a^i_t, a^{-i}_t, z_t\big)  + \sum_{t=1}^T  2 \beta_t \sigma_t \big( a^i_t, a^{-i}_t, z_t \big) \right]  \nonumber \\ 
& \leq  \mathbb{E}_{\substack{z_{\cdots T} \\ a^i_{\cdots T}}}  \underbrace{ \left[ \sum_{t=1}^T    \ucb_t \big( \pi(z_t), a^{-i}_t, z_t\big) -     \ucb_t \big(a^i_t, a^{-i}_t, z_t\big) \right]}_{\hat{R^i_c}(\pi, T)} + 4 \beta_T \sqrt{\gamma_T \lambda T} \label{eq:regret_ucb} \, ,
\end{align}
where we have used the definition of $\ucb_t(\cdot)$, $\beta_t$ set according to Lemma~\ref{lemma:conf_lemma}, and \cite[Lemma 5.4]{srinivas2009gaussian}. 

Now, we consider a generic sequence $\{\epsilon_t \}_{t=1}^T$ of noise realizations and proceed bounding the expected value of $\hat{R^i_c}(\pi, T)$ for any of such sequences. Moreover, as in \cite{neu2020} we will make use of a \emph{ghost sample} $z_0 \sim \zeta$ which is sampled from the contexts' distribution independently from the whole history $\mathcal{H}_T$ of the game.  Also, we now explicitly consider the adaptiveness of the adversary.
Using the law of total expectation, the expected value of $\hat{R^i_c}(\pi, T)$ can be rewritten as\looseness=-1

\begin{align}
& \mathbb{E}_{\substack{z_{\cdots T} \\ a^i_{\cdots T}}} \big[ \hat{R^i_c}(\pi, T) \big]=    \mathbb{E}_{\substack{z_{\cdots T} \\ a^i_{\cdots T}}} \left[ \sum_{t=1}^T    \ucb_t \big( \pi(z_t), f(\mathcal{H}_{t-1}), z_t\big) -     \ucb_t \big(a^i_t, f(\mathcal{H}_{t-1}), z_t\big) \right]   \nonumber \\
& =   \mathbb{E}_{\substack{z_{\cdots T} \\ a^i_{\cdots T}}} \left[ \sum_{t=1}^T   \mathbb{E}_{z_t, a^i_t} \big[ \ucb_t\big( \pi(z_t), f(\mathcal{H}_{t-1}), z_t\big) -     \ucb_t\big(a^i_t, f(\mathcal{H}_{t-1}), z_t\big) \mid \mathcal{H}_{t-1} \big] \right] \nonumber \\  
& = \mathbb{E}_{\substack{z_{\cdots T} \\ a^i_{\cdots T}}} \left[ \sum_{t=1}^T   \mathbb{E}_{z_t} \big[ \ucb_t\big( \pi(z_t), f(\mathcal{H}_{t-1}), z_t\big) -   \sum_{a \in \mathcal{A}^i} p_t(z_t)[a]\cdot   \ucb_t\big(a, f(\mathcal{H}_{t-1}), z_t\big) \mid \mathcal{H}_{t-1} \big] \right] \nonumber \\
& = \mathbb{E}_{\substack{z_{\cdots T} \\ a^i_{\cdots T}}} \left[ \sum_{t=1}^T   \mathbb{E}_{z_0} \big[ \ucb_t\big( \pi(z_0), f(\mathcal{H}_{t-1}), z_0\big) -   \sum_{a \in \mathcal{A}^i} p_t(z_0)[a]\cdot   \ucb_t\big(a, f(\mathcal{H}_{t-1}), z_0\big) \mid \mathcal{H}_{t-1} \big] \right] \nonumber   \\
& =  \mathbb{E}_{\substack{z_{\cdots T} \\ a^i_{\cdots T} \\ z_0 }} \left[ \sum_{t=1}^T   \ucb_t\big( \pi(z_0), f(\mathcal{H}_{t-1}), z_0\big) -   \sum_{a \in \mathcal{A}^i} p_t(z_0)[a]\cdot   \ucb_t\big(a, f(\mathcal{H}_{t-1}), z_0\big) \right]  \label{eq:regret_z0} 
\end{align}
The second equality follows by the law of total expectation. The third equality holds since, conditioned on the history $\mathcal{H}_{t-1}$, $a_t^i$ is distributed according to $p_t(z_t)$. The fourth equality follows since $z_t$ and $z_0$ have the same distribution and the functions $\ucb_t(\cdot)$ and $p_t(\cdot)$ do not depend on the realization of $z_t$.
The last inequality is obtained by applying again the law of total expectation.

At this point, we can apply Lemma~\ref{lem:MW_rule} considering the sequence of functions $g_1,\ldots, g_T$ with $g_\tau(\cdot) = \ucb_\tau(\cdot, f(\mathcal{H}_{\tau-1}), z_0)$ for $\tau = 1,\ldots, T$ and noting that, for each $z_0$, $p_t(z_0)$ computed using the MW rule \eqref{eq:GPMWrule_v2} corresponds to the distribution computed according to rule~\eqref{eq:MW_rule_explicited} for each $t=1,\ldots, T$. Therefore, \eqref{eq:regret_z0} implies that 
\begin{equation*}
    \mathbb{E}_{\substack{z_{\cdots T} \\ a^i_{\cdots T}}} \big[ \hat{R^i_c}(\pi, T) \big] \leq \frac{\log K}{\eta_T} + \frac{\sum_{t=1}^T \eta_t}{8} \,.
\end{equation*}
Finally, the theorem statement is obtained substituting the bound above in \eqref{eq:regret_ucb} and considering the constant learning rate $\eta_t = \sqrt{8 \log (K) /T}$.
\end{proof}

\section{Supplementary Material for Section~\ref{sec:equilibria_and_efficiency}}
\subsection{Proof of Proposition~\ref{prop:convergence_c-CCE} (Finite-time approximation of c-CCEs)}
\begin{proof}
After $T$ rounds of the contextual game, consider a generic player $i$.
By definition of contextual regret, see  \eqref{eq:contextual_regret}, we have
\begin{equation}\label{eq:c-CCE_eq1}
 \frac{1}{T} \sum_{t=1}^T  r^i (a_t,a^{-i}_t, z_t ) \geq   \frac{1}{T} \sum_{t=1}^T r^i (\pi (z_t) , a^{-i}_t, z_t ) -   \frac{R_c^i(T)}{T} \qquad \forall \pi \in \Pi^i \, .
\end{equation}
Let now $\rho_T$ be the empirical policy up to time $T$, defined as in Section~\ref{subsec:cce}. Then, it is not hard to verify that the above cumulative rewards for player $i$ can be written as 
\begin{align*}
 & \frac{1}{T} \sum_{t=1}^T  r^i (a_t,a^{-i}_t, z_t ) =  \frac{1}{T} \sum_{t=1}^T \; \mathop{\mathbb{E}}_{  \mathbf{a}  \sim \rho_T(z_t)} r(\mathbf{a}, z_t) \, , \\
 &  \frac{1}{T} \sum_{t=1}^T  r^i (\pi(z),a^{-i}_t, z_t ) =  \frac{1}{T} \sum_{t=1}^T   \; \mathop{\mathbb{E}}_{  \mathbf{a}  \sim \rho_T(z_t)} r(  \pi(z), a^{-i}, z_t) \,.
\end{align*}
Therefore, \eqref{eq:c-CCE_eq1} becomes: 
\begin{equation}\label{eq:c-CCE_eq2}
\frac{1}{T} \sum_{t=1}^T  \; \mathop{\mathbb{E}}_{  \mathbf{a}  \sim \rho_T(z_t)} r(\mathbf{a}, z_t) \geq  \frac{1}{T} \sum_{t=1}^T  \; \mathop{\mathbb{E}}_{  \mathbf{a}  \sim \rho_T(z_t)} r(  \pi(z), a^{-i}, z_t)  -   \frac{R_c^i(T)}{T} \qquad \forall \pi \in \Pi^i \, .
\end{equation}
Note that this is precisely the condition of $\epsilon$-c-CCE (see Definition~\ref{def:c-CCE}) for player $i$. The final result is then simply obtained by considering the player with the highest regret. 
\end{proof}

\subsection{An alternative notion of c-CCE for stochastic contexts}\label{app:stochastic_cCCE}

In Section~\ref{sec:equilibria_and_efficiency} we defined the notion of c-CCE (Definition~\ref{def:c-CCE}) for a contextual game described by an arbitrary sequence of contexts $z_1,\ldots, z_T$. 
In this section, we consider the case in which contexts are stochastic samples from the same distribution $\zeta$, i.e., $z_t\sim \zeta$ for all $t$. 
In such a case, the following alternative notion of c-CCE can be defined by considering the \emph{expected} context realization 
(rather than considering the time-averaged game as in Definition~\ref{def:c-CCE}).

\begin{definition}\label{def:c-CCE_stochastic}
Consider a contextual game and assume contexts are sampled i.i.d. from distribution $\zeta$. A contextual coarse-correlated equilibrium for stochastic contexts (c-$\zeta$-CCE) is a policy $\rho : \mathcal{Z} \rightarrow \Delta^{|\bm{\mathcal{A}}|}$ mapping contexts to distributions over $\bm{\mathcal{A}}$ such that:\looseness=-1
\begin{equation}
    \mathop{\mathbb{E}}_{  z \sim \zeta}   \; \mathop{\mathbb{E}}_{  \mathbf{a}  \sim \rho(z)} r^i\big(\mathbf{a}, z\big) \geq    \mathop{\mathbb{E}}_{  z \sim \zeta}  \; \mathop{\mathbb{E}}_{  \mathbf{a}  \sim \rho(z)}  r^i \big(\pi(z), a^{-i},z \big) \quad \forall \pi \in  \Pi^i , \quad  \forall i = 1, \ldots, N  \,.
    \end{equation} 
    Moreover, $\rho$ is an $\epsilon$-c-$\zeta$-CCE if the above inequality is satisfied up to an $\epsilon\in \mathbb{R}_+$ accuracy.
\end{definition}

Similarly to Proposition~\ref{prop:convergence_c-CCE}, the following proposition shows that, in this specific setting, c-$\zeta$-CCEs can also be approached whenever players minimize their contextual regrets. 

\begin{proposition}[Asymptotic and finite-time convergence to  c-$\zeta$-CCEs]\label{prop:convergence_c-zeta-CCE} Consider a contextual game and assume contexts are sampled i.i.d. from distribution $\zeta$. Let $\rho_T$ be the empirical policy at round $T$. Then, as $T\rightarrow \infty$, if players have vanishing contextual regrets, $\rho_T$ converges to a c-$\zeta$-CCE almost surely.
Moreover, after $T$ game rounds, let $R^i_c(T)$'s denote the players' contextual regrets, $\delta \in (0,1)$, and assume $\mathcal{Z}$ is finite. Then, with probability at least $1-\delta$, $\rho_T$ is a $\epsilon$-c-$\zeta$-CCE with
\begin{equation*}
    \epsilon \leq  2 \: \sqrt{ \frac{\log( |\mathcal{Z}|\cdot |\bm{\mathcal{A}}|)}{2} + \frac{\log(2/ \delta)}{2T}} + \max_{i \in \{1, \ldots, N\} } \frac{R^i_c(T)}{T}  \, .
\end{equation*}
\vspace{-1.0em}
\end{proposition}

Compared to CCEs (and c-CCEs), c-$\zeta$-CCEs can be approximated in finite time only with high-probability and with an extra approximation factor of $\mathcal{O}(\log(|\mathcal{Z}|\: |\mathcal{A}|) + \log(1/\delta)/T)$. Intuitively, this is because the empirical distribution of observed contexts needs to concentrate around the true contexts distribution $\zeta$. We recover asymptotic convergence to c-$\zeta$-CCEs since such distribution converges to $\zeta$ with probability 1.

\begin{proof}
By definition of contextual regret, see  \eqref{eq:contextual_regret}, for each player $i$
\begin{equation}\label{eq:c-zeta-CCE_eq1}
 \frac{1}{T} \sum_{t=1}^T  r^i (a_t,a^{-i}_t, z_t ) \geq   \frac{1}{T} \sum_{t=1}^T r^i (\pi (z_t) , a^{-i}_t, z_t ) -   \frac{R_c^i(T)}{T} \qquad \forall \pi \in \Pi^i \, .
\end{equation}
Let $\zeta_T$ be the empirical distribution of observed contexts. Moreover, let $\rho_T$ be the empirical policy up to time $T$, defined in Section~\ref{subsec:cce}. Then, following the same steps of Proof of Proposition~\ref{prop:convergence_c-CCE}:
\begin{align*}
 & \frac{1}{T} \sum_{t=1}^T  r^i (a_t,a^{-i}_t, z_t ) =  \mathop{\mathbb{E}}_{  z \sim \zeta_T}   \; \mathop{\mathbb{E}}_{  \mathbf{a}  \sim \rho_T(z)} r(\mathbf{a}, z) \, , \\
 &  \frac{1}{T} \sum_{t=1}^T  r^i (\pi(z),a^{-i}_t, z_t ) =  \mathop{\mathbb{E}}_{  z \sim \zeta_T}   \; \mathop{\mathbb{E}}_{  \mathbf{a}  \sim \rho_T(z)} r(  \pi(z), a^{-i}, z) \,.
\end{align*}
Therefore, \eqref{eq:c-zeta-CCE_eq1} rewrites as 
\begin{equation}\label{eq:c-zeta-CCE_eq2}
\mathop{\mathbb{E}}_{  z \sim \zeta_T}   \; \mathop{\mathbb{E}}_{  \mathbf{a}  \sim \rho_T(z)} r(\mathbf{a}, z) \geq  \mathop{\mathbb{E}}_{  z \sim \zeta_T}   \; \mathop{\mathbb{E}}_{  \mathbf{a}  \sim \rho_T(z)} r(  \pi(z), a^{-i}, z)  -   \frac{R_c^i(T)}{T} \qquad \forall \pi \in \Pi^i \, .
\end{equation}

As $T \rightarrow \infty$, $\zeta_T \rightarrow \zeta$ as contexts are i.i.d. samples from $\zeta$. Moreover, if players use no-regret strategies, $R_c^i(T)/T \rightarrow 0$ for $i=1,\ldots N$ and hence the above inequality implies that $\rho_T$ converges to a c-$\zeta$-CCE (see Definition~\ref{def:c-CCE_stochastic}).

For finite $T$, the above inequality resembles the desired c-$\zeta$-CCE condition, with the difference that the outer expectations are taken with respect to the empirical contexts' distribution $\zeta_T$ instead of the true one. 
To cope with this, we show that such expectations indeed concentrate, up to some accuracy, around the expectations with respect to the true distribution $\zeta$. More precisely, we show that with probability at least $1-\delta$
\begin{equation}\label{eq:concentration_bound_cce}
   \left| \mathop{\mathbb{E}}_{  z \sim \zeta_T}  f(z, \rho_T) \: - \: \mathop{\mathbb{E}}_{  z \sim \zeta}   f(z, \rho_T) \right| \leq  \: \sqrt{ \frac{\log(|\mathcal{Z}|\cdot |\bm{\mathcal{A}}|)}{2} + \frac{\log(2/ \delta)}{2T}} \, ,
\end{equation}
where $f(z, \rho_T) = \mathop{\mathbb{E}}_{  \mathbf{a}  \sim \rho_T(z)} r(\mathbf{a}, z)$. Moreover, the same condition holds for  $f(z, \rho_T) = \mathop{\mathbb{E}}_{  \mathbf{a}  \sim \rho_T(z)} r(  \pi(z), a^{-i}, z)$ for each $\pi \in \Pi^i$. Combined with \eqref{eq:c-zeta-CCE_eq2}, this implies that for each player $i$ and each $\pi \in \Pi^i$, with probability $1-\delta$,
\begin{equation*}
\mathop{\mathbb{E}}_{  z \sim \zeta}   \; \mathop{\mathbb{E}}_{  \mathbf{a}  \sim \rho_T(z)} r(\mathbf{a}, z) \geq  \mathop{\mathbb{E}}_{  z \sim \zeta}   \; \mathop{\mathbb{E}}_{  \mathbf{a}  \sim \rho_T(z)} r(  \pi(z), a^{-i}, z)  - 2\:  \sqrt{ \frac{\log(|\mathcal{Z}|\cdot |\bm{\mathcal{A}}|)}{2} + \frac{\log(2/ \delta)}{2T}} -  \frac{R_c^i(T)}{T} \, ,
\end{equation*}
which would prove Proposition~\ref{prop:convergence_c-zeta-CCE}. 

It remains to show \eqref{eq:concentration_bound_cce}. For a \emph{given} policy $\rho: \mathcal{Z} \rightarrow \Delta^{|\bm{\mathcal{A}}|}$, a straightforward application of Hoeffding's inequality~\cite{hoeffding1963} shows that for any $\epsilon>0$
\begin{equation}\label{eq:application_of_hoeffding}
    \mathbb{P}\left[  \left| \mathop{\mathbb{E}}_{  z \sim \zeta_T}  f(z, \rho) \: - \: \mathop{\mathbb{E}}_{  z \sim \zeta}   f(z, \rho) \right| > \epsilon \right] \leq 2 \exp \big(-2 T \epsilon^2 \big) \, ,
\end{equation}
 where we have used the fact that $f(z,\rho) = \mathbb{E}_{\mathbf{a}\sim \rho(z)} r^i(\mathbf{a}, z) \in [0,1]$ and that $z_1, \ldots, z_T$ are i.i.d. sampled from $\zeta$. Unfortunately, we cannot apply the condition above directly to the empirical policy $\rho_T$, since it is not fixed a-priori, but is computed as a function of the realized samples $z_1, \ldots, z_T$. However, we can consider the set $\mathcal{P}_T$ of all the possible empirical policies $\rho: \mathcal{Z} \rightarrow \Delta^{|\bm{\mathcal{A}}|}$ resulting from $T$ rounds of the repeated game. 
Note that each of such policies is uniquely defined by the sequence $\{a_t^i,a_t^{-i} z_t \}_{t=1}^T$ of revealed contexts and actions played up to round $T$. Therefore, $\mathcal{P}_T$ is a finite set of cardinality $|\mathcal{P}_T| = (|\mathcal{Z}|\cdot |\bm{\mathcal{A}}|)^T$. 
Hence, it holds
 \begin{align}\nonumber
    \mathbb{P}\left[  \left| \mathop{\mathbb{E}}_{  z \sim \zeta_T}  f(z, \rho_T) \: - \: \mathop{\mathbb{E}}_{  z \sim \zeta}   f(z, \rho_T) \right| > \epsilon \right] & \leq \mathbb{P}\left[ \sup_{\rho \in \mathcal{P}_T} \left| \mathop{\mathbb{E}}_{  z \sim \zeta_T}  f(z, \rho) \: - \: \mathop{\mathbb{E}}_{  z \sim \zeta}   f(z, \rho) \right| > \epsilon \right]  \\ 
    & = \mathbb{P}\left[ \bigcup_{\rho \in \mathcal{P}_T}  \left\{ \left| \mathop{\mathbb{E}}_{  z \sim \zeta_T}  f(z, \rho) \: - \: \mathop{\mathbb{E}}_{  z \sim \zeta}   f(z, \rho) \right| > \epsilon \right\} \right]  \nonumber \\
    & \leq  |\mathcal{P}_T| \:  \mathbb{P}\left[  \left| \mathop{\mathbb{E}}_{  z \sim \zeta_T}  f(z, \rho) \: - \: \mathop{\mathbb{E}}_{  z \sim \zeta}   f(z, \rho) \right| > \epsilon \right]  \nonumber \\
    & \leq 2 |\mathcal{P}_T| \: \exp \big(-2 T \epsilon^2 \big) \, . \nonumber
\end{align}
The first equality holds since, given a set of random variables $x_1, \ldots, x_n$, asking that $\sup_i x_i > \epsilon$ is equivalent to asking that at least one of the $x_i$'s is greater than $\epsilon$. The second inequality is a standard probability union bound, while the last inequality follows by \eqref{eq:application_of_hoeffding}. 
This proves \eqref{eq:concentration_bound_cce} after setting the right hand side equal to $\delta$ and substituting $\big|\mathcal{P}_{T}\big| = (|\mathcal{Z}|\cdot |\bm{\mathcal{A}}|)^{T}$.
\end{proof}

\subsection{Proof of Proposition~\ref{prop:efficiency} (Convergence to approximate efficiency)} 

\begin{proof}
 For ease of notation, let $\pi_\star^1, \ldots, \pi_\star^N$ be the optimal policies that solve \eqref{eq:social_optimum}, so that $\mathrm{OPT} = \frac{1}{T}\sum_{t=1}^T \Gamma \big(\pi_\star^1(z_t), \ldots, \pi_\star^N(z_t), z_t \big)$. Using the definitions of contextual regret and $(\lambda, \mu)$-smoothness, the sum of cumulative rewards can be lower bounded as:  
\begin{align*}
& \frac{1}{T}\sum_{t=1}^T \sum_{i=1}^N r^i  ( a^i_t, a^{-i}_t, z_t ) \\ 
& \geq \frac{1}{T}\sum_{t=1}^T \sum_{i=1}^N r^i  \big( \pi_\star^i(z_t), a^{-i}_t, z_t \big) \quad - \sum_{i=1}^N \frac{R^i_c(T)}{T} \\ 
& \geq  \frac{1}{T}\sum_{t=1}^T \Big[ \lambda(z_t) \cdot \Gamma \big(\pi_\star^1(z_t), \ldots, \pi_\star^N(z_t), z_t \big)  - \mu(z_t) \cdot \Gamma( a_t^1, \ldots, a_t^N, z_t )  \Big] \quad - \sum_{i=1}^N \frac{R^i_c(T)}{T}  \\ 
& \geq   \bar{\lambda} \cdot \mathrm{OPT}  -  \: \bar{\mu} \cdot  \frac{1}{T}\sum_{t=1}^T  \Gamma( a_t^1, \ldots, a_t^N, z_t )  \quad  - \sum_{i=1}^N \frac{R^i_c(T)}{T}   \, .
\end{align*}
In the first inequality we have used the definition of contextual regret (see \eqref{eq:contextual_regret}) with respect to policy $\pi_\star^i$ for each player (note that $\pi_\star^i$ is not necessarily the optimal policy in hindsight for player $i$). In the second inequality we have used the fact that the game is $\big(\lambda(z_t), \mu(z_t)\big)$-smooth at each time $t$ and applied condition \eqref{eq:smoothness} with outcomes $\mathbf{a}_1 = (a_t^1, \ldots, a_t^N)$ and $\mathbf{a}_2 = \big(\pi_\star^1(z_t), \ldots, \pi_\star^N(z_t) \big)$. The last inequality follows from the definition of $\bar{\lambda}, \bar{\mu}$, and $\mathrm{OPT}$. 

At this point, note that $ \frac{1}{T}\sum_{t=1}^T  \Gamma( a_t^1, \ldots, a_t^N, z_t ) \geq \frac{1}{T}\sum_{t=1}^T \sum_{i=1}^N r^i  ( a^i_t, a^{-i}_t, z_t )$ since by definition of social welfare $\Gamma(\mathbf{a},z) \geq \sum_{i=1}^N r^i(\mathbf{a}, z)$ for all $(\mathbf{a},z)$. Then, the above inequalities imply that
\begin{equation*}
 \frac{1}{T}\sum_{t=1}^T  \Gamma( a_t^1, \ldots, a_t^N, z_t)  \geq \bar{\lambda} \cdot \mathrm{OPT}  -  \: \bar{\mu} \cdot  \frac{1}{T}\sum_{t=1}^T  \Gamma( a_t^1, \ldots, a_t^N, z_t )  \quad  - \sum_{i=1}^N \frac{R^i_c(T)}{T} \, , 
\end{equation*}
which after rearranging yields the desired result.
\end{proof}

\section{Contextual Traffic Routing - Experimental Setup}\label{app:routing}

In this section we describe the experimental setup of the contextual traffic routing game of Section~\ref{sec:experiments}. We consider the traffic network of Sioux-Falls, a directed graph with $24$ nodes and $76$ edges and use the game model of \cite{sessa2019noregret}. Data from \cite{leblanc1975,website_transp_test} include node coordinates and capacities $C_e\in \mathbb{R}_+$ of each network's edge $e= 1,\ldots, 76$. Moreover, data also include the units (e.g., cars) that need to be sent from any node to any other node in the network, for a total of $528$ distinct origin-destination pairs. Hence, we let $N= 528$ be the number of agents in the network and assume, at every round, each agent~$i$ needs to send $d^i$ units from origin node $O^i$ to destination node $D^i$. In order to send these units, each agent can choose one of the $K = 5$ shortest routes between $O^i$ and $D^i$. We let $x_t
^i \in \mathbb{R}^{76}$ represent the route chosen by agent $i$ at round $t$, where $x_t^i[e] = d_i$ if edge $e$ belongs to such route, and $x_t^i[e] = 0$ otherwise.
Moreover, we let $x_t
^{-i} = \sum_{j \neq i} x_t^j$ represent the routes chosen by the rest of the agents.\looseness=-2

At each round, the network displays different capacities (network's capacities represent the contextual information of the game) which are observed by the agents and should be used to choose better routes, depending on the circumstances. This is different from the game model of \cite{sessa2019noregret} where network capacities are assumed constant. We let the context vector $z_t\in \mathbb{R}_+^{76}$ represent the network's capacities at round $t$, and assume each $z_t$ is i.i.d. sampled from a static distribution $\zeta$. The contexts distribution $\zeta$ is generated as follows. We let $\mathcal{Z}$ be a set of $10$ randomly generated capacity profiles $z$ where $z[e]$ is uniformly distributed in $[0, 1.2\cdot C_e]$ for $e = 1,\ldots, 76$. Then, we let $\zeta$ be the uniform distribution over $\mathcal{Z}$.\looseness=-1 

Given context $z_t$ and routes $x_t^i$, $x_t^{-i}$, the reward of each agent $i$ is: 
\begin{equation*}
    r^i(x_t^i, x_t^{-i},z_t) = - \sum_{e =1}^{76} x_t^i[e] \cdot t_e(x_t^i + x_t^{-i}, z_t[e]) \, ,
\end{equation*}
where $t_e(\cdot)$ is the traveltime function of edge $e$ (i.e., the relation between number of units traversing edge $e$ and the time needed to traverse it). Such functions are unknown to the agents, and according to \cite{leblanc1975,website_transp_test} are defined by the Bureau of Public Roads (BPR) congestion model: 
\begin{equation*}
t_e(x, z) = f_e\cdot  \Bigg(1  + 0.15\Big( \frac{x}{z}\Big)^4 \Bigg)\, ,
\end{equation*}
where $f_e\in\mathbb{R}_+$ is the free-flow traveltime of edge $e$ (also provided by the network's data).
At the end of each round, hence, we quantify the congestion of each edge $e$ with the quantity $0.15 ((x_t^i[e] + x_t^{-i}[e])/z_t[e])^4$.

To run our experiments, we estimate upper and lower bounds on the agents' rewards by sampling $10'000$ random contexts and game outcomes, and feed such bounds to the agents so that rewards can be scaled in the $[0,1]$ interval. Moreover, at each round agents receive a noisy measurement of their rewards, with noise standard deviation $\sigma$ set to $0.1\%$. 
To run \textsc{RobustLinExp3}~\cite[Theorem 1]{neu2020} we set learning rate $\eta = 0.3$ and exploration parameter $\gamma = 0.2$ (we observe worse performance when setting them to their theoretical values). For \textsc{GP-MW} we use the composite kernel $k(x_t^i, x_t^{-i}, z_t) = k_1(x_t^i) * k_2(x_t^i + x_t^{-i})$ used also in \cite{sessa2019noregret}, while for \textsc{c.GP-MW} the kernel $k(x_t^i, x_t^{-i}, z_t) = k_1(x_t^i) * k_2((x_t^i + x_t^{-i})/z_t)$, where $k_1$ is a linear kernel and $k_2$ is a polynomial kernel of degree 4. However, we observe similar performance when polynomials of different degrees are used or when $k_2$ is the widely used SE kernel. Kernel hyperparameters are optimized offline over 100 random datapoints and kept fixed. We set $\eta_t$ according to Theorems~\ref{thm:thm1} and \ref{thm:thm2}, and confidence level $\beta_t = 2.0$ (theoretical values for $\beta_t$ are found to be overly conservative, as also observed in \cite{srinivas2009gaussian,sessa2019noregret}).\looseness=-1

\end{document}